\documentclass[a4paper,USenglish,cleveref, autoref, thm-restate, nolineno]{socg-lipics-v2021}

\hideLIPIcs  

\usepackage{tikz,pgfplots,mathtools}
\usetikzlibrary{matrix, arrows, calc, patterns, intersections, decorations.markings, shapes.geometric, chains, decorations.pathmorphing,pgfplots.fillbetween}
\usepackage[textsize=scriptsize]{todonotes}
\usepackage[ruled,noend,noline]{algorithm2e}

\newcommand{\R}{\mathbb R}

\newcommand{\cl}{\textrm{cl}}
\newcommand{\C}{\mathcal{C}}
\newcommand{\Ss}{\mathcal{S}}

\newcommand{\M}{\mathcal{M}}

\DeclareMathOperator{\lfs}{lfs}
\DeclareMathOperator{\Neigh}{Neigh}
\DeclareMathOperator{\closest}{closest}
\DeclareMathOperator{\clComp}{clComp}
\DeclareMathOperator*{\argmin}{arg\,min}
\DeclareMathOperator{\CompNeigh}{CompNeigh}
\DeclareMathOperator{\reach}{reach}

\newtheorem*{propmain}{\cref{Prop:main}}
\newtheorem*{propfinish}{\cref{Prop:finish}}

\newtheorem*{abdisc}{\cref{Lem:abDisc}}
\newtheorem*{Xab}{\cref{Lem:Xab}}

\def\centerarc[#1](#2)(#3:#4:#5)
{ \draw[#1] ($(#2)+({#5*cos(#3)},{#5*sin(#3)})$) arc (#3:#4:#5); }

\title{Tighter Bounds for Reconstruction from $\epsilon$-samples} 

\author{H{\aa}vard {Bakke Bjerkevik}}{TU Graz, Austria}{bjerkevik@tugraz.at}{https://orcid.org/0000-0001-9778-0354}{}

\authorrunning{H.\,B. Bjerkevik} 

\Copyright{H{\aa}vard Bakke Bjerkevik} 



\category{} 

\relatedversion{} 

\funding{The author is supported by the Austrian Science Fund (FWF) grant number P 33765-N.}

\acknowledgements{The author would like to thank Michael Kerber for insights into the size of Delaunay triangulations, and Stefan Ohrhallinger and Scott A. Mitchell for answering my questions about the state of the art of curve and surface reconstruction.}

\begin{document}

\maketitle

\begin{abstract}
We show that reconstructing a curve in $\R^d$ for $d\geq 2$ from a $0.66$-sample is always possible using an algorithm similar to the classical {\sc NN-Crust} algorithm. Previously, this was only known to be possible for $0.47$-samples in $\R^2$ and $\frac{1}{3}$-samples in $\R^d$ for $d\geq 3$.
In addition, we show that there is not always a unique way to reconstruct a curve from a $0.72$-sample; this was previously only known for $1$-samples. We also extend this non-uniqueness result to hypersurfaces in all higher dimensions.
\end{abstract}

\section{Introduction}
The main problem considered in this paper is that of \emph{curve reconstruction}. Given a (finite) set of points $\Ss$ in $\R^d$, we assume that this is a subset of a union $\C$ of closed curves, and we want to reconstruct $\C$ knowing only $\Ss$.
Reconstructing $\C$ exactly from a finite set of points is unfeasible, so we restrict the problem to finding the graph $G_\C(\Ss)$ on $\Ss$ induced by $\C$: there is an edge in $G_\C(\Ss)$ between two points in $\Ss$ if you can walk from one to the other along $\C$ without meeting another point of $\Ss$.

To do this, one needs an assumption on $\Ss$ and $\C$. Some work on curve reconstruction and similar problems uses global assumptions for instance related to the maximum curvature \cite{attali1998r, bernardini1997sampling, de1994computational, niyogi2008finding, stelldinger2008topologically, stelldinger2009provably}.
A weakness of this approach is that it may force you to sample the whole curve densely even if just a small portion of it has large curvature.
An influential paper by Amenta, Bern and Eppstein \cite{amenta1998crust} introduced the {\sc Crust} algorithm along with a local sampling condition allowing the sampling density to vary depending on the local distance to the \emph{medial axis} of $\C$.
To be precise, they guarantee correct reconstruction for any \emph{$\epsilon$-sampled} curve in the plane whenever $\epsilon<0.252$.
The condition that a curve is $\epsilon$-sampled is weaker the larger $\epsilon$ is, so we would like to guarantee correct reconstruction for $\epsilon$-sampled curves for as large an $\epsilon$ as possible.

There followed a number of papers seeking to improve the sampling conditions of \cite{amenta1998crust}: Dey and Kumar \cite{dey1999simple} introduced {\sc NN-Crust} (NN = nearest neighbor), which allows curves in higher-dimensional space, and prove that correct reconstruction is guaranteed for $\epsilon < \frac{1}{3}$; Lenz \cite{lenz2006sample} defines a family of algorithms of which {\sc NN-Crust} is a special case and conjectures that $\epsilon\leq 0.48$ is sufficient for correctness in another special case;
and Ohrhallinger et al. \cite{ohrhallinger2016curve} introduce {\sc HNN-Crust}, proving correct reconstruction for $\epsilon < 0.47$, and also for $\rho < 0.9$, where $\rho$ is a \emph{reach}-based parameter that is related to (but different from) the parameter $\epsilon$. It is shown in \cite[Observation 6]{amenta1998crust} that correct reconstruction cannot be guaranteed for $\epsilon\geq 1$.

In addition, there have been several papers improving on \cite{amenta1998crust} in other ways: Gold \cite{gold1999crust} simplified the {\sc Crust} algorithm; Dey et al. \cite{dey2000curve} gave an algorithm allowing open curves; and Dey and Wenger \cite{dey2002fast} considered curves with corners.
Finally, we mention that \cite{althaus2001traveling} ties the $\epsilon$-sampling condition to a completely different approach to curve reconstruction by showing that a solution of the traveling salesman problem on the sample points gives a correct reconstruction from an $\epsilon$-sample for $\epsilon<0.1$.
For further references, we refer to the recent survey of Ohrhallinger et al. \cite{ohrhallinger20212d} on curve reconstruction in the plane.

Moving up to higher dimensions, one can consider the problem of \emph{submanifold reconstruction} \cite{abdelkader2018sampling,chazal2008smooth, dey2006curve, dey2008critical, niyogi2008finding}.
Instead of working with samples of a curve, one assumes that the points are sampled from a submanifold in $\R^d$ for $d\geq 3$; the case of surfaces in $\R^3$ is of particular interest.
While this is not the main focus of the paper, we note that this problem is important from a practical point of view; see for instance \cite{berger2017survey} for a survey covering the literature related to 3D scannings with imperfections.
So far, the results using $\epsilon$-sampling have been much weaker for surface reconstruction than for curve reconstruction. For $d=3$, correct surface reconstruction is only known to be possible to guarantee for $\epsilon\leq 0.06$ \cite{amenta2000simple}.

\subsection{Our contributions}

The question we study is: For which $\epsilon$ is it possible to guarantee correct curve reconstruction using an $\epsilon$-sample? Despite the popularity of $\epsilon$-sampling as a sampling condition in the literature and the body of work aiming to weaken sampling conditions, there is still a large gap between the $\epsilon$ for which we know that reconstruction is always possible and the $\epsilon$ for which we know that it is not always possible:
For any $\epsilon\in (0.47,1)$, it is as far as the author knows an open question if it is possible to guarantee correct reconstruction of a curve (or union of curves) in $\R^2$ using an $\epsilon$-sample.
For curves in $\R^d$, $d\geq 3$, the same is true for $\epsilon\in (\frac{1}{3},1)$. We improve this situation drastically in both ends. First we describe algorithms that guarantee correct reconstruction for $\epsilon=0.66$ for all $d\geq 2$. Algorithm \ref{Alg:main_alg} runs in $O(n^2)$ for any fixed $d$, and Algorithm \ref{Alg:main_alg_planar} runs in $O(n\log n)$ for $d=2$.
While we have not implemented our algorithms, we believe that the speed of Algorithm \ref{Alg:main_alg_planar} in practice is comparable to that of the algorithms in \cite{dey1999simple} and \cite{ohrhallinger2016curve} because of their similarities.

Secondly, we give an example demonstrating that one cannot in general guarantee correct reconstruction using $0.72$-samples for any $d\geq 2$.
Thus, the interval of $\epsilon$ for which it is unknown if an $\epsilon$-sample is enough for reconstruction is reduced from $(0.47,1)$ (or $(\frac{1}{3},1)$ for $d\geq 3$) to $(0.66,0.72)$.

By a straightforward generalization, we use our example to prove that a $0.72$-sample is not in general enough to guarantee correct reconstruction of a manifold of any dimension. We do not show any positive results in higher dimensions, but we hope that since we do not put any restriction on the ambient dimension of the set of samples, our ideas can be useful also for reconstruction of higher-dimensional manifolds.

A serious alternative to the $\epsilon$-sampling condition is the $\rho$-sampling condition of \cite{ohrhallinger2016curve}. The authors of \cite{ohrhallinger2016curve} argue that $\epsilon$-sampling with $\epsilon\leq 0.47$ requires more sample points than what $\rho$-sampling does.
With our new bounds on $\epsilon$, the situation changes somewhat. An in-depth discussion of the relationship between $\epsilon$-sampling and $\rho$-sampling is beyond the scope of this paper (as is the question of whether the two sampling conditions can be combined in a way that exploits the advantages of both of them), but we study some instructive examples in \cref{Sec:rho-sampling}.
To summarize, $\rho$-sampling seems to do better for curves with slowly changing curvature, while $\epsilon$ does better in some examples with rapidly changing curvature. Both our upper and lower bounds for $\epsilon$ help us understand the relative strengths of $\epsilon$- and $\rho$-sampling.

We begin by introducing necessary definitions and notation in \cref{Sec:defs}, before we prove the main theorem in \cref{Sec:proof}. In \cref{Sec:counterexample}, we show that correct reconstruction from $0.72$-samples is not always possible, and we finish off by generalizing the example to higher dimensions in \cref{Sec:hypersurface}.

\section{Definitions and notation}
\label{Sec:defs}

Throughout most of the paper, we work with a finite, disconnected union $\C$ of closed curves in $\R^d$ for some fixed $d\geq 2$, and a finite subset $\Ss$ of $\C$. We will call the elements of $\Ss$ \emph{sample points}.
By a closed curve, we mean the image of an injective map from the circle.
Sometimes it will be convenient to fix an orientation of (a connected component of) $\C$.
The notation $a\to b$ means that we have chosen an orientation of a connected component of $\C$ containing $a,b\in \Ss$ and that by starting at $a$ and moving along $\C$ following this orientation, the next element of $\Ss$ one encounters is $b$. We use the shorthand $a\to b\to c$ when we mean $a\to b$ and $b\to c$.
For $p,q$ in the same connected component of $\C$, we define $[p,q]$ as $\{p\}$ if $p=q$, and as the image of any injective path from $p$ to $q$ that is consistent with the orientation of $\C$ if $p\neq q$.
We define $[a,b)$, $(a,b]$ and $(a,b)$ similarly depending on whether $a$ and/or $b$ are included or not. By a \emph{midpoint of $[a,b]$} we mean a point $p\in [a,b]$ with $d(p,a)=d(p,b)$, where $d(x,y)$ denotes Euclidean distance.

If $(a\to) b\to c$ or $c\to b (\to a)$, we say that ($a$,) $b$ and $c$ are \emph{consecutive}. We define $G_\C(\Ss)$ as the graph on $\Ss$ with an edge between $a$ and $b$ if and only if $a$ and $b$ are consecutive.

For $X\subset \R^d$, let $d(x,X)\coloneqq\inf_{y\in X} d(x,y)$.
The \emph{medial axis} $\M$ \cite{Blum} is the set of points in $\R^d$ that do not have a unique closest point in $\C$. For $p\in \C$, the \emph{local feature size} $\lfs(p)$ is defined as $d(p,\M)$.
For $\epsilon>0$, we say that $\Ss\subset \C$ is an \emph{$\epsilon$-sample} (of $\C$) if for all $p\in \C$, $d(p,\Ss) < \epsilon \lfs(p)$. Note that being an $\epsilon$-sample is a stronger condition the smaller $\epsilon$ is.
Throughout the paper we will assume that $\Ss$ is an $\epsilon$-sample, but our assumptions on $\epsilon$ will vary.

We define $\cl\colon \R^d\setminus \M \to \C$ by letting $\cl(x)$ be the point in $\C$ closest to $x$; i.e., $\cl(x) = \arg\min_{p\in \C} d(x,p)$.
It follows immediately from the definition of $\M$ that $\cl$ is well-defined. We prove that $\cl$ is continuous in \cref{Lem:clCont}.

We use the notation $B_x(r)$ for the closed ball with radius $r$ centered at $x\in \R^d$. For $x,y\in\R^d$, the closed line segment from $x$ to $y$ is denoted by $\overline{xy}$.

We often restrict our attention to a plane $\Pi\subset \R^d$, which we identify with $\R^2$. This way, we can associate canonical coordinates $(x,y)$ to each point $p\in \Pi$.

\section{Proof that $0.66$-samples allow reconstruction}
\label{Sec:proof}

This section is devoted to giving a proof of the main theorem:
\begin{theorem}
\label{Thm:main}
Let $\C$ be a union of closed curves in $\R^d$ for some $d\geq 2$, and let $\Ss$ be a $0.66$-sample of $\C$ containing $n$ points. Given $\Ss$ as input, {\sc NN-compatible} and {\sc Compatible-crust} both compute $G_\C(\Ss)$. The former runs in $O(n^2)$, and for $d=2$, the latter runs in $O(n\log n)$.
\end{theorem}
The algorithms are rather simple, and are similar to the previous \textsc{Crust}-type algorithms.
To be specific, {\sc Compatible-crust} borrows the idea from \cite{amenta1998crust} of only selecting edges from the Delaunay triangulation\footnote{For an introduction to Delaunay triangulations in the plane, see \cite[Chapter 9]{dutch}.},
and both algorithms use the idea from \cite{dey1999simple} of including an edge between each sample point and its nearest neighbor (called ``closest'' in the algorithms) in addition to the nearest neighbor satisfying some condition related to the angle between the resulting two edges (called ``clComp'' in the algorithms).
The new ingredient in our algorithm is that we require triples of consecutive points to be \emph{compatible} (see \cref{Fig:compatible}), which is a different criterion than those used in previous algorithms.
We define this compatibility property in \cref{Subsec:compatible}.
This criterion has the advantage over criteria used in previous papers in that it is the optimal local criterion for when a triple of points can be consecutive:
If a triple is not compatible, it cannot be consecutive, while if it is compatible, there is a curve passing through the three points that does not violate the sampling condition locally.
It will be clear from the definition that checking if a triple $(a,b,c)\in \Ss^3$ is compatible can be done in constant time.
The separation into two algorithms is done to optimize the running time: For $d=2$, computing the Delaunay triangulation saves us time, while for $d\geq 3$, a more straightforward approach is at least as efficient in the worst case.

\begin{algorithm}
\internallinenumbers
\DontPrintSemicolon
\KwIn{$0.66$-sample $\Ss\subset \R^d$ of $\C$ for $d\geq 2$}
\KwOut{$G_\C(\Ss)$}
Initialize $G\gets \{\}$\;
\ForEach{$x\in \Ss$}{
$\closest \gets \argmin_{y\in \Ss\setminus\{x\}}\{d(x,y)\}$\;
$\CompNeigh\gets \{y\in \Ss\mid (\closest,x,y)\text{ is compatible}\}$\;
$\clComp \gets \argmin_{y\in \CompNeigh}\{d(x,y)\}$\;
$G\gets G\cup \{\{x,\closest\},\{x,\clComp\}\}$\;
}
\Return{G}\;
\caption{{\sc NN-compatible}}
\label{Alg:main_alg}
\end{algorithm}

In {\sc NN-compatible}, we run through the for-loop $n$ times. Each line in the loop can be executed in $O(n)$, which gives a total running time of $O(n^2)$.

\begin{algorithm}
\internallinenumbers
\DontPrintSemicolon
\KwIn{$0.66$-sample $\Ss\subset \R^d$ of $\C$ for $d\geq 2$}
\KwOut{$G_\C(\Ss)$}
Compute the $1$-skeleton $D_1(\Ss)$ of a Delaunay triangulation of $\Ss$.\;
Initialize $G\gets \{\}$\;
\ForEach{$x\in \Ss$}{
$\Neigh \gets$ the set of vertices in $D_1(\Ss)$ adjacent to $x$\;
$\closest \gets \argmin_{y\in \Neigh}\{d(x,y)\}$\;
$\CompNeigh\gets \{y\in \Neigh\mid (\closest,x,y)\text{ is compatible}\}$\;
$\clComp \gets \argmin_{y\in \CompNeigh}\{d(x,y)\}$\;
$G\gets G\cup \{\{x,\closest\},\{x,\clComp\}\}$\;
}
\Return{G}\;
\caption{{\sc Compatible-crust}}
\label{Alg:main_alg_planar}
\end{algorithm}

Computing a Delaunay triangulation in the plane can be done in $O(n\log n)$ \cite[Theorem 9.12]{dutch}.
The total number of edges in $D_1(\Ss)$ is $O(n)$, so the sum of the sizes of all the $\Neigh$ over all $x\in \Ss$ is $O(n)$. Thus, the total running time of the for-loop is $O(n)$. This gives a running time for {\sc Compatible-crust} of $O(n\log n+n) = O(n\log n)$ for $d=2$.
For $d\geq 3$, the Delaunay triangulation may have a size as large as $\Theta(n^{\lceil d/2 \rceil})$ \cite[Chapter 27.1]{toth2017handbook}, in which case {\sc Compatible-crust} does not do better than {\sc NN-compatible} for $d\in \{3,4\}$ and does worse for $d\geq 5$.

It remains to be proved that the algorithms output $G_\C(\Ss)$. Since {\sc Compatible-crust} restricts itself to the set of edges of the Delaunay triangulation, we need to know that this set contains the edges of $G_\C(\Ss)$.
In the planar case, this is proved in \cite[Lemma 11]{amenta1998crust}. We extend the result to higher ambient dimensions in \cref{Cor:Delaunay}.

Finally, we need to prove that the closest and ``closest compatible'' neighbors to a sample point are indeed the adjacent vertices in $G_\C(\Ss)$.
As the proof is rather long and technical, we devote a full section to it, which we split into three subsections: In \cref{Subsec:basic}, we prove a sequence of lemmas about the local behavior of $\Ss$ and $\C$. Then, in \cref{Subsec:restrictions}, we prove lower bounds on the angle between certain triples of points on $\C$; in particular, \cref{Lem:abDisc} implies that consecutive triples of points have to be compatible.
Lastly, in \cref{Subsec:compatible}, we use the results from the first two subsections to prove that the edges constructed by the algorithms are indeed exactly the edges in $G_\C(\Ss)$.

\subsection{Basic observations about $\Ss$ and $\C$}
\label{Subsec:basic}

Recall that $\Ss$ is assumed to be an $\epsilon$-sample of $\C$. In this subsection, we assume $\epsilon\leq 1$. Later, we will restrict $\epsilon$ to smaller values and state our assumptions on $\epsilon$ explicitly in each case.

For $p\in \C$, define $d_p = d(p,\Ss)$.
By definition of $\cl$ and $\epsilon$-sample, $\cl$ is defined in $B_p\left(\frac{d_p}{\epsilon}\right)$. Since we assume $\epsilon\leq 1$, $\cl$ is in particular defined in $B_p(d_p)$. We will use the following lemma throughout the paper without referring to it explicitly.
\begin{lemma}
\label{Lem:clCont}
$\cl$ is continuous.
\end{lemma}
\begin{proof}
Let $x\in \R^d\setminus \M$, and let $x_1,x_2,\dots$ be a sequence of points in $\R^d\setminus \M$ that converges to $x$. To show that $\cl$ is continuous, it is enough to show that the image of the sequence under $\cl$ converges to $\cl(x)$.
Let $y$ be an accumulation point in $\C$ of the sequence $\cl(x_1), \cl(x_2), \dots$, which exists by compactness of $\C$. Then $d(x,y)\leq d(x,y')$ for any $y'\in \C$, so $y=\cl(x)$. Thus, $\cl(x)$ is the only accumulation point of $\cl(x_1), \cl(x_2), \dots$, so by compactness of $\C$, the sequence converges to $\cl(x)$.
\end{proof}

\begin{lemma}
\label{Lem:I}
Let $x\in\R^d$ and $q\in\C$ be such that $\overline{xq}$ does not intersect the medial axis. Let $p=\cl(x)$. Then the interior of $B_x(d(x,q))$ contains either $[p,q)$ or $(q,p]$.
\end{lemma}
\begin{proof}
By continuity of $\cl$ and connectedness of $\overline{xq}$, $\cl(\overline{xq})$ must contain either $[p,q]$ or $[q,p]$. Suppose the former. Then for any $z\in [p,q)$, $z=\cl(i)$ for some $i\in \overline{xq}$. Thus, \[d(x,q)=d(x,i)+d(i,q)>d(x,i)+d(i,z)\geq d(x,z).\] The statement follows, and the argument for $[q,p]$ is exactly the same.
\end{proof}

\begin{lemma}
\label{Lem:d_p_closest}
Let $a\to b$ and $p\in (a,b)$. Then $d_p=\min\{d(p,a),d(p,b)\}$, and $d_p<d(p,s)$ for all $s\in \Ss\setminus \{a,b\}$.
\end{lemma}
\begin{proof}
Suppose $s\notin \{a,b\}$ is a point in $\Ss$ minimizing the distance to $p$, so $d_p=d(p,s)$. Then $B_p(d(p,s))=B_p(d_p)$ and thus $\overline{ps}$ does not intersect the medial axis.
Since $\cl(p)=p$, \cref{Lem:I} (with $x=p$ and $q=s$) shows that the interior of $B_p(d_p)$ contains either $a$ or $b$, which is a contradiction, as then either $d(p,a)$ or $d(p,b)$ would be smaller than $d(p,s)$. Thus, $d_p$ is equal to either $d(p,a)$ or $d(p,b)$.
\end{proof}
As a step in proving the correctness of {\sc Compatible-crust}, we need to show that for $a\to b$, there is an edge between $a$ and $b$ in the Delaunay triangulation of $\Ss$.
Since we do not assume that $\Ss$ is in general position, we do not know that there is a unique Delaunay triangulation of $\Ss$.
Still, we know that if there is a closed ball $B$ such that $B\cap \Ss=\{a,b\}$, then any Delaunay triangulation of $\Ss$ has an edge between $a$ and $b$. In the special case of curves in the plane, the following was proved in \cite[Lemma 11]{amenta1998crust}.
\begin{corollary}
\label{Cor:Delaunay}
Let $a\to b$. Then there is an edge between $a$ and $b$ in any Delaunay triangulation of $\Ss$.
\end{corollary}
\begin{proof}
Let $p$ be a midpoint on $[a,b]$. By \cref{Lem:d_p_closest}, $B_p(d_p)\cap \Ss = \{a,b\}$, so there is an edge between $a$ and $b$ in the Delaunay triangulation of $\Ss$.
\end{proof}

For $x,y\in \R^d$, let $E(x,y)$ be the set of points in $\R^d$ that are equidistant from $x$ and $y$.

\begin{lemma}
\label{lem:unique2}
Let $b\in \Ss$, let $a\neq b$ be in the same connected component of $\C$ as $b$, let $p$ be either the midpoint on $[a,b]$ or equal to $a$, and assume $d_p=d(p,b)$.
Then for every $x\in B_p\left(\frac{d_p}{\epsilon}\right)\cap E(a,b)$,
\begin{itemize}
\item[(i)] $\cl(x) \in (a,b)$,
\item[(ii)] $(a,b)\subset B_x(d(x,b))$.
\end{itemize}
\end{lemma}
\begin{proof}
(i): Let $B= B_p\left(\frac{d_p}{\epsilon}\right)$, and let $m$ be the midpoint on $[a,b]$. If $p=a$, then by \cref{Lem:I}, $m\in B$. Trivially, $m\in B$ also holds if $p=m$.
Since $\Ss$ is an $\epsilon$-sample, $B$ does not intersect the medial axis, so $\cl:B\to \C$ is well-defined. Clearly, $\cl(m) = m$, and $a,b\notin \cl(B\cap E(a,b))$, as $d(a,x)=d(b,x)$ for every $x\in E(a,b)$. Since $\cl$ is continuous and $B\cap E(a,b)$ connected, we get that $\cl(B\cap E(a,b))\subset (a,b)$.

(ii): Since $\overline{xb}\subset B_p\left(\frac{d_p}{\epsilon}\right)$, \cref{Lem:I} tells us that $[\cl(x),b)$ is in the interior of $B_x(d(x,b))$ (since $a\in (b,\cl(x)]$ is not in the interior of $B_x(d(x,b))=B_x(d(x,a))$), and so must $(a,\cl(x)]$ by a symmetric argument.
\end{proof}

\begin{definition}
\label{Def:Xab}
For $a\neq b\in \R^d$, let $X(a,b)$ be the set of $x$ such that $d(x,a) = d(x,b) = \frac{d(a,b)}{\epsilon\sqrt{4-\epsilon^2}}$, and let $U(a,b) = \bigcap_{x\in X(a,b)}B_x\left(\frac{d(a,b)}{\epsilon\sqrt{4-\epsilon^2}}\right)$, which is equal to $\bigcap_{x\in X(a,b)}B_x\left(d(x,a)\right)$.
\end{definition}

\begin{figure}
\centering
\begin{tikzpicture}[scale=2.5]
\coordinate (a) at (0,-1){};
\node[left] at (a){$a$};
\coordinate (b) at (0,0){};
\node[left] at (b){$b$};
\coordinate (x) at (-0.628,-0.5){};
\node[right] at (x){$x$};
\coordinate (y) at (0.628,-0.5){};
\node[left] at (y){$y$};
\coordinate (p) at (-0.05,-0.5){};
\node[right] at (p){$p$};
\node[right] at (0.57,-1){$B_p\left(\frac{d_p}{\epsilon} \right)$};
\foreach \X in {a,b,x,y,p}
\draw[color=black,fill=black] (\X) circle (.02);
\centerarc[black](x)(-80:80:0.8025);
\centerarc[black](y)(100:260:0.8025);
\draw[color=black] (p) circle (0.761);
\begin{scope}
\clip (0,-.5) circle (0.5);
\draw[fill, opacity=.2, even odd rule] (x) circle (0.8025) (y) circle (0.8025) (0,-.5) circle (0.6);
\end{scope}
\draw[thick] (0.1,0.4) to [out=250,in=80] (b) to [out=260,in=90] (p) to [out=270,in=85] (a) to [out=265,in=80] (-0.1,-1.4);
\node[right] at (0.1,0.4){$\C$};
\end{tikzpicture}
\caption{The planar case with $X(a,b) = \{x,y\}$. The shaded area is $U(a,b)$ and contains $[a,b]$ by \cref{Lem:Xab} (iii). By \cref{Lem:Xab} (ii), $B_p\left(\frac{d_p}{\epsilon} \right)$ contains $X(a,b)$, where $p$ is the midpoint on $[a,b]$. \label{Fig:U}}
\end{figure}
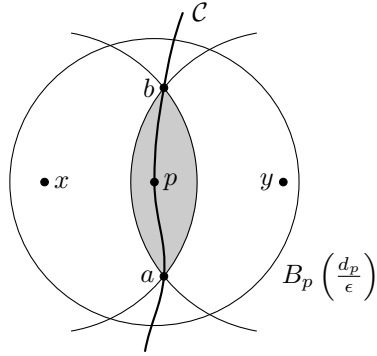

\begin{lemma}
\label{Lem:Xab}
Let $a\to b$.
\begin{itemize}
\item[(i)] Let $p'\in E(a,b)\cap \partial U(a,b)$, and let $x$ be the point in $X(a,b)$ maximizing the distance to $p'$. Then $d(p',a)=\epsilon d(p',x)$, $d(p',x)=d(a,x)$ and $2\angle axp' = \angle axb$.
\item[(ii)] Let $p$ be the midpoint of $[a,b]$. Then $X(a,b)\subset B_p\left(\frac{d_p}{\epsilon} \right)$.
\item[(iii)] $(a,b)\subset U(a,b)$.
\end{itemize}
\end{lemma}
See \cref{Fig:U} for an illustration of (ii) and (iii).
We prove the lemma in \cref{Subsec:proof_Xab}.

\subsection{Restrictions of angles between points on $\C$}
\label{Subsec:restrictions}

With help from the results of the previous subsection, we now prove results that essentially limit the curvature of $\C$ locally.

\begin{proposition}
\label{Prop:main}
Let $\epsilon\leq 0.765$, and let $a\to b\to c$ with $p\in (a,b)$ and $d(p,b)\leq d(p,a)$. Then for any $x$ such that $d(x,p)=d(x,b)=\frac{d_p}{\epsilon}$, $(b,c]\cap B_x\left(\frac{d_p}{\epsilon}\right) =\emptyset$.
\end{proposition}
The rough idea of the proof is to assume there is a $c'\in (b,c]\cap B_x\left(\frac{d_p}{\epsilon}\right)$ and consider a line segment $\overline{xm}$, where $x$ satisfies the conditions in the lemma and $m$ is the midpoint on $\overline{bc'}$.
One can show that $\cl$ is defined on $\overline{xm}$, that $\cl(x)\in (p,b)$, and that $\cl(m)\in (b,c')$ and derive that $\cl(\overline{xm})$ is disconnected, which is a contradiction by continuity of $\cl$. We give the full details in \cref{Subsec:proof_main}.

\begin{corollary}
\label{Cor:141}
Let $\epsilon\leq 0.66$, let $a\to b\to c$, and let $p$ be the midpoint of $[a,b]$ and $q\in (b,c]$. Then
\[
\angle pbq> 70.73^\circ + \arccos\left(0.33\frac{d(q,b)}{d(p,b)}\right).
\]
In particular, if $d(p,b)\geq d(q,b)$, then $\angle pbq>141^\circ$.
\end{corollary}
\begin{proof}
We restrict our attention to a plane containing $p$, $b$ and $q$ and assume without loss of generality that $p=(0,-1)$, $b=(0,0)$ and that $q$ is not to the left of the $y$-axis.
By \cref{Prop:main}, $q$ cannot be in the disc $D$ with radius $\frac{1}{\epsilon}$ with $p$ and $b$ on the boundary and center $x$ to the right of the $y$-axis.
Under this condition, we have $\angle pbq>\angle pbq'$, where $q'$ is on the boundary of $D$ above the $x$-axis and $d(q',b)=d(q,b)$.
As illustrated in \cref{Fig:141a}, $\cos \angle pbx = \frac{1/2}{1/\epsilon} \leq 0.33$. Similarly, $\cos \angle xbq' = \frac{d(q',b)/2}{1/\epsilon} \leq 0.33d(q',b)$.
Since $\arccos$ is decreasing, we get
\begin{align*}
\angle pbq &> \angle pbq'\\
&= \angle pbx + \angle xbq'\\
&\geq \arccos(0.33)+\arccos(0.33d(q',b))\\
&> 70.73^\circ+\arccos(0.33d(q',b)).
\end{align*}
If we do not assume $d(p,b)=1$, we have to replace $d(q',b)$ with $\frac{d(q',b)}{d(p,b)}$ in the last expression. Since $d(q',b)=d(q,b)$, this yields the wanted inequality. If $d(p,b)\geq d(q,b)$, then this lower bound is weakest when $d(q,b)=d(p,b)$. In this case the right-hand side is $> 141.46^\circ$.
\end{proof}

\begin{figure}[t!]
\centering
\begin{subfigure}[t]{0.3\textwidth}
\centering
\begin{tikzpicture}[scale=2.5]
\coordinate (p) at (0,-1){};
\node[left] at (p){$p$};
\coordinate (b) at (0,0){};
\node[left] at (b){$b$};
\coordinate (x) at (1.43,-0.5){};
\node[right] at (x){$x$};
\coordinate (q) at (0.538,0.725){};
\node[above] at (q){$q'$};
\foreach \X in {p,b,x,q}
\draw[color=black,fill=black] (\X) circle (.03);
\draw (b) to (p) to (x) to (q) to (b) to (x) to (0,-0.5);
\draw (0,-0.42) to (0.08,-0.42) to (0.08,-0.5);
\centerarc[black](x)(110:210:1.5152);
\centerarc[black](b)(270:340.73:0.12);
\centerarc[black](b)(-19.27:53.4:0.1);
\node at (0.11,-0.17){$\phi$};
\node at (0.17,0.05){$\psi$};
\node at (0.75,-0.14){$\epsilon^{-1}$};
\end{tikzpicture}
\caption{If $d(p,b)=1$, then $\cos \phi = \frac{1/2}{1/\epsilon}$. \label{Fig:141a}}
\end{subfigure}\hspace{2cm}
\begin{subfigure}[t]{0.3\textwidth}
\centering
\begin{tikzpicture}[scale=4]
\coordinate (a) at (0,-1){};
\node[left] at (a){$a$};
\coordinate (b) at (0,0){};
\node[left] at (b){$b$};
\coordinate (x) at (0.628,-0.5){};
\node[right] at (x){$x$};
\coordinate (c) at (0.859,0.2685){};
\node[above] at (c){$c'$};
\foreach \X in {a,b,x,c}
\draw[color=black,fill=black] (\X) circle (.018);
\draw (b) to (a) to (x) to (c) to (b) to (x) to (0,-0.5);
\draw (0,-0.44) to (0.06,-0.44) to (0.06,-0.5);
\centerarc[black](x)(60:230:0.8025);
\centerarc[black](b)(270:321.46:0.12);
\centerarc[black](b)(17.35:-38.54:0.1);
\node at (0.07,-0.15){$\phi$};
\node at (0.15,-0.02){$\psi$};
\node at (0.45,-0.18){$\frac{1}{\epsilon\sqrt{4-\epsilon^2}}$};
\end{tikzpicture}
\caption{If $d(a,b)=1$, then $\cos \phi = \frac{\epsilon\sqrt{4-\epsilon^2}}{2}$. \label{Fig:141b}}
\end{subfigure}
\caption{\label{Fig:141}}
\end{figure}

\begin{lemma}
\label{Lem:abDisc}
Let $\epsilon\leq 0.765$, and let $a\to b\to c$. Then $(b,c]\cap B_x(d(x,a))=\emptyset$ for all $x\in X(a,b)$.
\end{lemma}
This proof is similar to that of \cref{Prop:main}; see \cref{Subsec:proof_abdisc} for the details.

\begin{definition}
We call a triple $(a,b,c)$ of sample points \emph{compatible} if $c\notin B_x(d(x,b))$ for all $x\in X(a,b)$ and $a\notin B_y(d(y,b))$ for all $y\in X(b,c)$.
\end{definition}
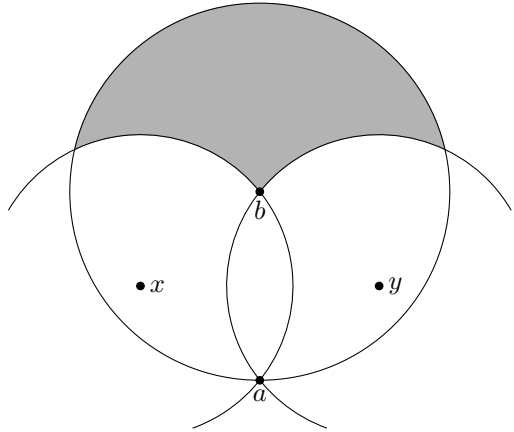
\begin{figure}
\centering
\begin{tikzpicture}[scale=2.5]
\coordinate (a) at (0,-1){};
\node[below] at (a){$a$};
\coordinate (b) at (0,0){};
\node[below] at (b){$b$};
\coordinate (x) at (-0.628,-0.5){};
\node[right] at (x){$x$};
\coordinate (y) at (0.628,-0.5){};
\node[right] at (y){$y$};
\foreach \X in {a,b,x,y}
\draw[color=black,fill=black] (\X) circle (.02);
\draw[color=black] (b) circle (1);
\centerarc[black](x)(-70:150:0.8025);
\centerarc[black](y)(30:250:0.8025);
\begin{scope}
\clip (b) to (0.975,0.224) to (0.975,1) to (-0.975,1) to (-0.975,0.224);
\draw[fill, opacity=.3, even odd rule] (b) circle (1) (x) circle (0.8025) (y) circle (0.8025);
\end{scope}
\end{tikzpicture}
\caption{The planar case with $X(a,b) = \{x,y\}$. If $d(b,c)\leq d(a,b)$, then $(a,b,c)$ is compatible if and only if $c$ is in the shaded area. \label{Fig:compatible}}
\end{figure}
See \cref{Fig:compatible}.
\cref{Lem:abDisc} then implies that if $a\to b\to c$, then $(a,b,c)$ is compatible.
\begin{lemma}
\label{Cor:102}
Let $\epsilon\leq 0.66$ and suppose $(a,b,c)$ is compatible. Then
\[
\angle abc> 51.45^\circ + \arccos\left(0.6231\frac{d(c,b)}{d(a,b)}\right).
\]
In particular, $\angle abc>102.9^\circ$.
\end{lemma}
\begin{proof}
We use an argument very similar to that in the proof of \cref{Cor:141}. 
We restrict our attention to the plane spanned by $a$, $b$ and $c$ and assume without loss of generality that $a=(0,-1)$, $b=(0,0)$ and that $c$ is not to the left of the $y$-axis. By definition of compatibility, $c\notin B_x(d(x,a))$, where $x$ is the element of $X(a,b)$ to the right of the $y$-axis.
Under this condition, we have $\angle abc>\angle abc'$, where $c'$ satisfies $d(c,b)=d(c',b)$ and is on the boundary of $B_x(d(x,a))$ above the $x$-axis.
We already determined the geometry of $\triangle axb$ in \cref{Fig:pythagoras}. Using this, we find that $\cos \angle abx = \frac{\epsilon\sqrt{4-\epsilon^2}}{2}$, as illustrated in \cref{Fig:141b}. Similar considerations show that $\cos \angle xbc' = \frac{\epsilon\sqrt{4-\epsilon^2}d(c',b)}{2d(a,b)}$.
We have $d(c',b)=d(c,b)$ by assumption, and $\frac{\epsilon\sqrt{4-\epsilon^2}}{2}\leq \frac{0.66\sqrt{4-0.66^2}}{2} < 0.6231$. Since $\arccos$ is decreasing, this yields
\begin{align*}
\angle abc &> \angle abc'\\
&= \angle abx+\angle xbc'\\
&> \arccos(0.6231) + \arccos\left(0.6231\frac{d(c,b)}{d(a,b)}\right)
\end{align*}
and then the wanted inequality follows from $\arccos(0.6231)> 51.45^\circ$.

$(a,b,c)$ is compatible if and only $(c,b,a)$ is, so the inequality holds also if we switch $a$ and $b$. Thus, we can assume $d(a,b)\geq d(c,b)$. Under this assumption, the right-hand side is smallest when $d(c,b)=d(a,b)$. Thus,
\[
\angle abc> 2\arccos(0.6231)>102.9^\circ.\qedhere
\]
\end{proof}

\subsection{The closest compatible neighbors are the correct neighbors}
\label{Subsec:compatible}

In the runtime analysis of our algorithms, we stated that checking if a triple $(a,b,c)$ of points is compatible can be done in constant time.
Since we only need to consider the geometry of three fixed points, this is clear; to be precise, by arguments similar to those in the proof of \cref{Cor:102}, what we need to check is if
\[\angle abc > \arccos\left(\frac{0.66\sqrt{4-0.66^2}}{2}\right) + \arccos\left(\frac{0.66\sqrt{4-0.66^2}d(c,b)}{2d(a,b)}\right)\]
and the same with $a$ and $c$ switching places.

Recall that our algorithms construct edges from $b\in\Ss$ to $a$ and $c$, where $a$ is the closest point in $\Ss$ to $b$, and $c$ is the closest point in $\Ss$ to $b$ such that $(a,b,c)$ is compatible. ({\sc Compatible-crust} is restricted to the Delaunay neighbors, which by \cref{Cor:Delaunay} is not a problem.)
Since $b$ has exactly two adjacent vertices in $G_\C(\Ss)$, it is sufficient to prove that $a$, $b$ and $c$ are consecutive. This is exactly the statement of \cref{Prop:finish} below, which therefore finishes the proof of \cref{Thm:main}.

\begin{lemma}
\label{Lem:117}
Let $\epsilon\leq 0.66$ and $a\to b$, let $p$ be the midpoint on $[a,b]$, and let $c$ be a point on $\C\setminus [a,b]$ with $d(b,c)\leq d(a,b)$. Then $\angle pbc>117.3^\circ$.
\end{lemma}
\begin{proof}
\begin{figure}[t!]
\centering
\begin{subfigure}[t]{.42\textwidth}
\centering
\begin{tikzpicture}[scale=1.4]
\coordinate (p) at (0,-1){};
\node[left] at (p){$p$};
\coordinate (b) at (0,0){};
\node[left] at (b){$b$};
\coordinate (a) at (0,-2){};
\node[below] at (a){$a$};
\coordinate (c) at (1.788, 0.8968){};
\node[right] at (c){$x$};
\coordinate (q) at (-1.43,-0.5){};
\node[right] at (q){$q$};
\coordinate (qq) at (1.43,-0.5){};
\node[right] at (qq){$q'$};
\foreach \X in {a,p,b,c,q,qq}
\draw[color=black,fill=black] (\X) circle (.038);
\draw (p) to (b) to (c);
\draw (p) arc (-19.27:19.27:1.5152);
\draw (qq) circle (1.3454);
\draw (b) circle (2);
\draw (p) circle (1.51);
\begin{scope}
\clip(b) to (a) to (3,-2) to (c);
\clip(b) circle (2);
\draw[color=red,fill=red,opacity=0.4] (p) circle (1.51);
\draw[color=blue,fill=blue,opacity=0.4] (qq) circle (1.3454);
\end{scope}
\end{tikzpicture}
\caption{The discs $B_p\left(\frac{d(p,b)}{\epsilon}\right)$ (red) and $B_{q'}(d(q,q')-\epsilon^{-1})$ (blue) cover the relevant area except a small part close to $x$. \label{Fig:117a}}
\end{subfigure}\hspace{1.8cm}
\begin{subfigure}[t]{.42\textwidth}
\centering
\begin{tikzpicture}[scale=1.4]
\coordinate (p) at (0,-1){};
\node[left] at (p){$p$};
\coordinate (b) at (0,0){};
\node[left] at (b){$b$};
\coordinate (c) at (1.788, 0.8968){};
\node[right] at (1.74, 1){$x$};
\coordinate (q) at (-1.43,-0.5){};
\node[right] at (q){$q$};
\coordinate (qq) at (1.43,-0.5){};
\node[right] at (qq){$q'$};
\coordinate (z) at (1.244,-0.1351){};
\node[right] at (z){$z$};
\foreach \X in {p,b,c,q,qq,z}
\draw[color=black,fill=black] (\X) circle (.038);
\draw (p) to (b) to (c);
\draw (p) arc (-19.27:19.27:1.5152);
\draw (z) circle (1.18);
\draw (b) circle (2);
\draw (p) circle (1.51);
\end{tikzpicture}
\caption{The distance from $z$ to $B_q(\epsilon^{-1})$ is slightly smaller than $d(z,x)$. \label{Fig:117b}}
\end{subfigure}
\caption{}
\end{figure}

Assume $\angle pbc \leq 117.3^\circ$, and let us restrict ourselves to a plane containing $p,b,c$.
Without loss of generality, we can assume that $p=(0,-1)$, $b=(0,0)$, and that $c$ is not to the left of the $y$-axis.
Let $q$ be the point to the left of the $y$-axis such that $d(q,b)=d(q,p)=\epsilon^{-1}$, and let $q'$ be the reflection of $q$ across the $y$-axis.
By \cref{lem:unique2} (i), $\cl(q')\in (p,b)$, and by \cref{lem:unique2} (ii) (choose $a=p$ in the lemma), $(p,b)\subset B_q(\epsilon^{-1})$. It follows that $c\notin B_{q'}(d(q,q')-\epsilon^{-1})$.

We have two remaining possibilities under the assumptions $\angle pbc \leq 117.3^\circ$ and $d(b,c)\leq d(a,b)$:
\begin{itemize}
\item[(i)] $c\in B_b(d(a,b))\cap B_p\left(\frac{d(p,b)}{\epsilon}\right)$,
\item[(ii)] $c\in B_b(d(a,b))\setminus \left( B_p\left(\frac{d(p,b)}{\epsilon}\right)\cup B_{q'}(d(q,q')-\epsilon^{-1})\right)$,
\end{itemize}
To show that (i) is impossible, first assume that $c$ is below or on the line $l$ through $q$ and $q'$. If $\overline{cq}$ does not intersect $\overline{ab}$, let $I=\overline{cq}$.
Otherwise, let $I=\overline{cq'}$. $c$ is closer to any point on $I$ than $a$ is, so $a\notin \cl(I)$. Since no point on $I$ is above $l$, $b\notin \cl(I)$, as $p$ is always at least as close as $b$. But clearly, $\cl(c)=c$, and we have already observed that $\cl(q')\in (p,b)$, and $\cl(q)\in (p,b)$ holds for the same reason. Thus, $\cl(I)$ is disconnected, which contradicts the continuity of $\cl$.

If instead $c$ is above $l$, let $I=\overline{cq'}$ and use a similar argument with $a$ and $b$ exchanged.

Finally, we assume (ii), which is the case that requires the most care. Let $z=(1.244,-0.1351)$.
Some calculation shows that $z\in B_p(\epsilon^{-1})=B_p\left(\frac{d_p}{\epsilon}\right)$.
Let $I=\overline{zq'}$. Since $I\subset B_p\left(\frac{d_p}{\epsilon}\right)$, $I$ does not intersect the medial axis of $C$.

Let $x$ be the intersection of the ray from $b$ into the first quadrant with angle $\angle 117.3^\circ$ with the boundary of $B_b(2)$. As \cref{Fig:117a} illustrates, $c$ must be in an area close to $x$, and $x$ is the point in this area furthest away from $z$.
Some more calculation shows that $d(z,x)<1.18<d(z,q)-\epsilon^{-1}$; see \cref{Fig:117b}. This means that $z$ is closer to $c$ than to any point on $[p,b]$, since $[p,b]\subset B_q(\epsilon^{-1})$, as we have observed.
Thus, $\cl(z)\notin [p,b]$.
In addition, $\cl(q')\in (p,b)$ by \cref{lem:unique2} (i).
But all points on $I$ are closer to $c$ than to both $p$ and $b$ (it is enough to check the endpoints of $I$), so $p,b\notin \cl(I)$. Thus, $\cl(I)$ is disconnected, which is impossible, as $\cl$ is continuous.
\end{proof}

\begin{proposition}
\label{Prop:closest}
Let $\epsilon\leq 0.66$, and let $a$ be a sample point and $b$ a closest neighbor to $a$ among the other sample points. Then $a$ and $b$ are consecutive.
\end{proposition}
\begin{proof}
Suppose for a contradiction that $x$, $a$ and $y$ are consecutive and $b\notin \{x,y\}$. Let $p$ be the midpoint of $[x,a]$ and $q$ the midpoint of $[a,y]$. By \cref{Cor:141}, $\angle paq>141^\circ$, and by \cref{Lem:117}, both $\angle pab$ and $\angle qab$ are greater than $117.3^\circ$, as $d(x,a),d(y,a)\geq d(a,b)$. The sum of these angles is greater than $360^\circ$, which is impossible.
\end{proof}
\begin{proposition}
\label{Prop:finish}
Let $\epsilon\leq 0.66$. Let $b$ be a sample point, $a$ a closest sample point to $b$, and $c$ the closest sample point to $b$ such that $(a,b,c)$ is compatible. Then $a$, $b$ and $c$ are consecutive.
\end{proposition}
In particular, there is a unique closest point $c$ to $a$ such that $(a,b,c)$ is compatible.

The idea of the proof is as follows: We let $a$, $b$ and $c$ be as in the proposition, suppose there is a $c'\neq c$ such that $a$, $b$ and $c'$ are consecutive, let $q$ be the midpoint of $[b,c']$, and carefully pick a point $p\in [a,b]$. We get lower bounds on $\angle qbc$ and $\angle pbq$ by \cref{Lem:117} and \cref{Prop:main} depending on the distances from $q$, $c$ and $p$ to $b$. This gives an upper bound on $\angle cbp$, which leads to a contradiction by an argument similar to the one in the proof of \cref{Lem:117}. However, the proof is complicated by the degrees of freedom we have in choosing the distances from the various points to $b$. We give the details in \cref{Subsec:proof_finish}.

\section{Counterexample to curve reconstruction for $\epsilon=0.72$}
\label{Sec:counterexample}

In this section, we prove the following theorem, which says that correct curve reconstruction using $0.72$-samples is not in general possible, even in $\R^2$. Moreover, one cannot determine whether the (union of) curve(s) has more than one connected component, and the reconstruction problem remains impossible also under the assumption that the sample is taken from a single connected curve.
\begin{theorem}
\label{Thm:counterex}
There is a finite set $\Ss\subset \R^2$ that is a $0.72$-sample of $\C_1$, $\C_2$, $\C_3$ and $\C_4$, where $\C_1$ and $\C_2$ are connected closed curves and $\C_3$ and $\C_4$ are disconnected unions of closed curves, and $G_{\C_i}(\Ss)\neq G_{\C_j}(\Ss)$ for all $i\neq j$.
\end{theorem}

As we will construct subsets of the curves before we construct the complete curves, we extend the definition of $\epsilon$-sampling to unions of closed curves in the obvious way.

Let $a=(0,-1)$, $b=(0,0)$, $c=(-1.008,0.614)$, $d=(-1.008,1.614)$.
Let $S_1$ and $S_2$ be the two tangent circles with the same radius such that $a,b\in S_1$, $c,d\in S_2$ and the tangent point is the midpoint $q$ between $b$ and $c$. Let $\C$ be the union of the part of $S_1$ running from $a$ to $q$ through $b$ and the part of $S_2$ running from $q$ to $d$ through $c$.

Next, let $a',b',c',d'$ be the points, $S'_1, S'_2$ the circles and $\C'$ the curve we get by translating the whole construction horizontally to the right so that $d(b,c)=d(b,c')$; see \cref{Fig:cex1}. Let $T$ be the set of midpoints of $[a,b]$, $[b,c]$, $[c,d]$, $[a',b']$, $[b',c']$ and $[c',d']$.
\begin{figure}
\centering
\begin{tikzpicture}[scale=1.5]
\coordinate (a) at (0,-1){};
\node[left] at (a){$a$};
\coordinate (b) at (0,0){};
\node[left] at (b){$b$};
\coordinate (c) at (-1.008,0.614){};
\node[left] at (c){$c$};
\coordinate (d) at (-1.008,1.614){};
\node[left] at (d){$d$};
\coordinate (a') at (2.015,-1){};
\node[left] at (a'){$a'$};
\coordinate (b') at (2.015,0){};
\node[left] at (b'){$b'$};
\coordinate (c') at (1.008,0.614){};
\node[left] at (c'){$c'$};
\coordinate (d') at (1.008,1.614){};
\node[left] at (d'){$d'$};
\coordinate (q) at (-0.504,0.307){};
\coordinate (q') at (1.511,0.307){};
\foreach \X in {a,b,c,d,a',b',c',d'}
\node at (\X){$\bullet$};
\node at (-1.7,-0.5){$S_1$};
\node at (0.3,1.9){$S_2$};
\node at (-1.4,1.2){$\C$};
\node at (0.8,-0.5){$S_1'$};
\node at (2.7,1.2){$S_2'$};
\node at (2.4,-0.5){$\C'$};
\foreach \X in {a,a'}
\draw[thick] (\X) arc (-37.57:79.73:0.82);
\foreach \X in {d,d'}
\draw[thick] (\X) arc (142.43:259.73:0.82);
\foreach \X in {q,q'}{
\draw[dashed] (\X) arc (-100.27:142.43:0.82);
\draw[dashed] (\X) arc (79.73:322.43:0.82);
}
\end{tikzpicture}
\caption{}
\label{Fig:cex1}
\end{figure}

\begin{lemma}
\label{Lem:only_check_midpoints}
If $\{a,b,c,d,a',b',c',d'\}$ is not a $0.72$-sample of $\C\cup\C'$, then there is a $t\in T$ such that $B_t\left(\frac{d_t}{0.72}\right)$ intersects the medial axis of $\C\cup\C'$.
\end{lemma}
\begin{proof}
By definition, if $\{a,b,c,d,a',b',c',d'\}$ is not a $0.72$-sample of $\C\cup\C'$, then there is a $p\in\C\cup\C'$ such that $B_p\left(\frac{d_p}{0.72}\right)$ intersects the medial axis of $\C\cup\C'$.
Thus, it is enough to show that for every $p\in \C\cup\C'\setminus T$, there is a $t\in T$ such that $B_p\left(\frac{d_p}{0.72}\right)\subset B_t\left(\frac{d_t}{0.72}\right)$.

Let $p\in (x,y)\subset \C$ for some $x\to y$. We know that $d_p=d(p,x)$ or $d_p=d(p,y)$ by \cref{Lem:d_p_closest}. Suppose $d_p=d(p,y)$ ($d_p=d(p,x)$ is similar), and pick $p'\in (p,y)$.
Let $B=B_p\left(\frac{d_p}{0.72}\right)$ and $B'=B_{p'}\left(\frac{d_{p'}}{0.72}\right)$. If $B\nsubseteq B'$, there is a point on the ray from $p$ through $p'$ in $B'\setminus B$, which means that $\frac{d_p}{0.72} < d(p,p')+\frac{d_{p'}}{0.72}$, or equivalently
\[
\frac{d_p-d_{p'}}{d(p,p')} < 0.72.
\]
Observe that if we let $p'$ approach $p$, then $\frac{d_p-d_{p'}}{d(p,p')}$ approaches $\cos \angle p'py$; see \cref{Fig:approaching_cosine}.
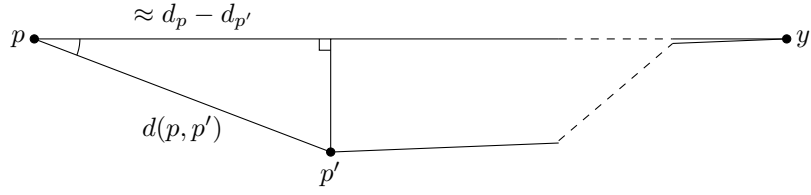
\begin{figure}
\centering
\begin{tikzpicture}[scale=3]
\coordinate (p) at (-1.3,0){};
\node[left] at (p){$p$};
\coordinate (pp) at (0,-0.5){};
\node[below] at (pp){$p'$};
\coordinate (y) at (2,0){};
\node[right] at (y){$y$};
\coordinate (o) at (0,0){};
\foreach \X in {p,y,pp}
\draw[color=black,fill=black] (\X) circle (.018);
\draw (o) to (pp) to (p) to (1,0);
\draw (pp) to (1,-0.46);
\draw[dashed] (1,-0.45) to (1.5,-0.02);
\draw[dashed] (1,0) to (1.5,0);
\draw (1.5,-0.02) to (y);
\draw (1.5,0) to (y);
\draw (-.05,0) to (-.05,-.05) to (0,-.05);
\node at (-0.6,0.1){$\approx d_p-d_{p'}$};
\node at (-0.65,-0.4){$d(p,p')$};
\centerarc[black](p)(-21.04:0:0.2);
\end{tikzpicture}
\caption{Assuming $d(p,p')\ll d(p,y)$, we have $\cos(\angle p'py)\approx \frac{d_p-d_{p'}}{d(p,p')}$. The dotted lines represent that we have collapsed a large part of the figure.}
\label{Fig:approaching_cosine}
\end{figure}
One can check that $\angle p'py<40^\circ$ for the possible $p$ and $p'$ in our example (by a large margin), while $\arccos(0.72)>43^\circ$. Thus,
\[
\arccos(0.72) > \arccos\left(\frac{d_p-d_{p'}}{d(p,p')}\right)
\]
for $p'$ sufficiently close to $p$,
so $0.72<\frac{d_p-d_{p'}}{d(p,p')}$, a contradiction. This proves that as $p$ moves along $\C$ or $\C'$ from a point in $T$ towards a point in $\{a,b,c,d\}$, the disc $B_p\left(\frac{d_p}{0.72}\right)$ decreases (in the sense that later discs are contained in earlier discs), proving the lemma.
\end{proof}

\begin{lemma}
\label{Lemma:8_points}
$\{a,b,c,d,a',b',c',d'\}$ is a $0.72$-sample of $\C\cup\C'$.
\end{lemma}
\begin{proof}
By \cref{Lem:only_check_midpoints}, what we need to show is that for any $t\in T$, $B_t\left(\frac{d_t}{0.72}\right)$ does not intersect the medial axis.

To reduce the problem, observe that we have symmetry around the midpoint between $b$ and $c'$, as $\overrightarrow{bc}=-\overrightarrow{c'b'}$ and $\overrightarrow{cd}=-\overrightarrow{b'a'}$. Thus, we can restrict ourselves to the midpoints $p$, $q$ and $r$ of $[a,b]$, $[b,c]$ and $[c,d]$, respectively. For the rest of the proof, assume $t\in \{p,q,r\}$.

We extend $\cl$ to a set-valued map from $\R^2$ to $\C\cup\C'$ by letting $\cl(p)$ be the set of points in $\C\cup\C'$ that minimize the distance to $p$. Let $m$ be a point such that $\cl(m)$ contains at least two points in $\C\cup\C'$, and let $x$ and $y$ be distinct points in $\cl(m)$.
We will show that for $t\in \{p,q,r\}$, $m\notin B_t\left(\frac{d_t}{0.72}\right)$.

There are the following cases to consider:
\begin{itemize}
\item $x,y\in \C$,
\item $x,y\in \C'$,
\item $x\in\C$ and $y\in \C'$.
\end{itemize}
In the first case, $m$ is on the medial axis of $\C$. This has two connected components: one is a curve starting at the center $s_1$ of $S_1$ and going leftwards and downwards from there, and the other is the mirror image through $q$ of the first one.
Because of symmetry, we only have to consider the first component. On this curve, $s_1$ minimizes the distance to $p$ and $q$, and $r$ is far away from the whole curve.
One can check that the radius of $S_1$ is greater than $0.82$, that $d_q = d(q,b) <0.59$ and that $d_q>d_p$. Thus,
\[
d(s_1,p)=d(s_1,q)>0.82 >\frac{d_q}{0.72}>\frac{d_p}{0.72},
\]
so we conclude that $B_t\left(\frac{d_t}{0.72}\right)$ does not intersect the medial axis of $\C$.

Next, we assume that $x,y\in \C'$. Then there is a point $m'$ on the line segment $\overline{mt}$ such that $\cl(m')$ intersects both $\C$ and $\C'$, so $m'$ is on the medial axis. If $m\in B_t\left(\frac{d_t}{0.72}\right)$, then $m'\in B_t\left(\frac{d_t}{0.72}\right)$, so we have reduced the second case to the third case.

\begin{figure}
\centering
\begin{tikzpicture}[scale=1.5]
\coordinate (a) at (0,-1){};
\coordinate (b) at (0,0){};
\coordinate (c) at (-1.008,0.614){};
\coordinate (d) at (-1.008,1.614){};
\coordinate (a') at (2.015,-1){};
\coordinate (d') at (1.008,1.614){};
\coordinate (p) at (0.17,-0.5){};
\node[left] at (p){$p$};
\coordinate (q) at (-0.504,0.307){};
\node[above] at (q){$q$};
\coordinate (r) at (-1.178,1.114){};
\node[left] at (r){$r$};
\coordinate (q') at (1.511,0.307){};
\foreach \X in {a,b,c,d,p,q,r}
\node at (\X){$\bullet$};
\node at (-1.3,1.5){$\C$};
\node at (2.4,-0.5){$\C'$};
\node at (-0.4,1.9){$l$};
\foreach \X in {a,a'}
\draw[thick] (\X) arc (-37.57:79.73:0.82);
\foreach \X in {d,d'}
\draw[thick] (\X) arc (142.43:259.73:0.82);
\draw (p) circle (0.731);
\draw (q) circle (0.82);
\draw (r) circle (0.731);
\draw (-0.606,1.894) to (1.61,-1.275);
\end{tikzpicture}
\caption{}
\label{Fig:cex2}
\end{figure}

Lastly, assume that $x\in\C$ and $y\in \C'$; see \cref{Fig:cex2}.
Let $l$ be the perpendicular bisector of $s_1$ and the center $s'_2$ of $S'_2$. If $y\in S'_2$, then $m$ is either on $l$ or to the right of $l$ (the latter can only happen if $x\in S_2$).
By numerical calculation, one can check that $l\cap B_t\left(\frac{d_t}{0.72}\right)=\emptyset$, so in this case, $m\notin B_t\left(\frac{d_t}{0.72}\right)$.
At the same time, if $y\in S_1'$, then $d(t,y)>\frac{2d_t}{0.72}$, so if $m\in B_t\left(\frac{d_t}{0.72}\right)$, then $y\notin S_1'$, as $t$ is closer to $m$ than $S_1'$ is.
\end{proof}

\begin{figure}
\centering
\begin{tikzpicture}[scale=2]
\node[left] at (-1.008,0.614){$c$};
\node[right] at (-1.008,1.614){$d$};
\node[left] at (-1,2.75){$e$};
\node[left] at (-1.008,3.614){$f$};
\coordinate (r) at (-1.178,1.114){};
\node[left] at (r){$r$};
\coordinate (s) at (-0.86,2.182){};
\node[left] at (s){$s$};
\begin{scope}[yscale=-2+1,xscale=-2+1,xshift=1cm,yshift=-0.614cm]
\coordinate (a) at (0,-1){};
\coordinate (b) at (0,0){};
\coordinate (ad) at (0,-2){};
\coordinate (add) at (0,-3){};
\coordinate (infl) at (0.005,-2.203){};
\coordinate (infl2) at (0,-2.293){};
\foreach \X in {a,b,ad,add,r,s}
\node at (\X){$\bullet$};
\draw (a) arc (-37.57:79.73:0.82);
\draw (a) arc (142.43:207.22:0.734);
\draw (infl) arc (-6.59:27.22:0.734);
\draw (infl) arc (173.41:180.55:0.734);
\draw (infl2) to (add);
\end{scope}
\begin{scope}[yscale=-1,xscale=-1,xshift=1cm,yshift=-0.614cm]
\draw (0.581,-1.447) circle (0.734);
\draw (-0.724,-2.119) circle (0.734);
\node at (0.581,-1.447){$S_3$};
\node at (-0.724,-2.119){$S_4$};
\end{scope}
\end{tikzpicture}
\caption{$\C$ is extended from $d$ through $e$ to $f$, where the tangent is vertical. \label{Fig:cexz}}
\end{figure}
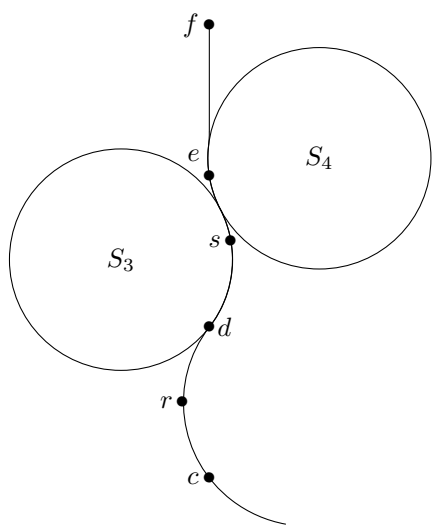

\begin{figure}
\centering
\begin{subfigure}{.5\textwidth}
\centering
\begin{tikzpicture}[scale=.7]
\node[left] at (0,-1){$a$};
\node[right] at (0,0.05){$b$};
\node[left] at (-1.008,0.614){$c$};
\node[right] at (-1.008,1.614){$d$};
\node[left] at (-0.93,2.914){$e$};
\node[left] at (-1.008,3.614){$f$};
\node[left] at (0,-3){$f'$};
\foreach \a in {0,1}{
\begin{scope}[xshift=\a*2.03cm]
\foreach \b in {0,1}{
\begin{scope}[yscale=-\b*2+1,xscale=-\b*2+1,xshift=\b*1cm,yshift=\b*-0.614cm]
\coordinate (a) at (0,-1){};
\coordinate (b) at (0,0){};
\coordinate (ad) at (0,-2){};
\coordinate (add) at (0,-3){};
\coordinate (infl) at (0.005,-2.203){};
\coordinate (infl2) at (0,-2.293){};
\foreach \X in {a,b,ad,add}
\node at (\X){$\bullet$};
\draw (a) arc (-37.57:79.73:0.82);
\draw (a) arc (142.43:207.22:0.734);
\draw (infl) arc (-6.59:27.22:0.734);
\draw (infl) arc (173.41:180.55:0.734);
\draw (infl2) to (add);
\end{scope}
}
\end{scope}
}
\begin{scope}[yscale=-1,xscale=-1,xshift=1cm,yshift=-0.614cm]
\end{scope}
\draw (add) arc (0:180:1.008);
\node at (0.712,4.326){$\bullet$};
\node at (0,4.621){$\bullet$};
\node at (-0.712,4.326){$\bullet$};
\draw[color=white] (0,-3.9) circle (0.01);
\end{tikzpicture}
\caption{$\C$ and $\C'$ are extended and tied together. \\ At $f$ and $f'$, the tangents of the curve are \\ vertical.}
\label{Fig:cex3}
\end{subfigure}%
\begin{subfigure}{.5\textwidth}
\centering
\begin{tikzpicture}[scale=.7]
\foreach \a in {0,1,2}{
\begin{scope}[xshift=\a*4.03cm]
\node at (0.712,4.326){$\bullet$};
\node at (0,4.621){$\bullet$};
\node at (-0.712,4.326){$\bullet$};
\draw[thick] (1.008,3.614) arc (0:180:1.008);
\end{scope}
}
\begin{scope}[xscale=1,yscale=-1,xshift=-1.008cm,yshift=-0.614cm]
\coordinate (a) at (0,-1){};
\coordinate (add) at (0,-3){};
\coordinate (infl) at (0.005,-2.203){};
\coordinate (infl2) at (0,-2.293){};
\draw[thick,color=red] (a) arc (-37.57:79.73:0.82);
\draw[thick,color=red] (a) arc (142.43:207.22:0.734);
\draw[thick,color=red] (infl) arc (-6.59:27.22:0.734);
\draw[thick,color=red] (infl) arc (173.41:180.55:0.734);
\draw[thick,color=red] (infl2) to (add);
\end{scope}
\begin{scope}[yscale=-1,xscale=-1,xshift=-9.06cm,yshift=-0.614cm]
\coordinate (a) at (0,-1){};
\coordinate (add) at (0,-3){};
\coordinate (infl) at (0.005,-2.203){};
\coordinate (infl2) at (0,-2.293){};
\draw[thick,color=blue] (a) arc (-37.57:79.73:0.82);
\draw[thick,color=blue] (a) arc (142.43:207.22:0.734);
\draw[thick,color=blue] (infl) arc (-6.59:27.22:0.734);
\draw[thick,color=blue] (infl) arc (173.41:180.55:0.734);
\draw[thick,color=blue] (infl2) to (add);
\end{scope}
\foreach \a in {0,1,...,4}{
\begin{scope}[xshift=\a*2.015 cm]
\foreach \b in {0,1}{
\begin{scope}[yscale=-\b*2+1,xscale=-\b*2+1,xshift=\b*1cm,yshift=\b*-0.614cm]
\coordinate (a) at (0,-1){};
\coordinate (b) at (0,0){};
\coordinate (ad) at (0,-2){};
\coordinate (add) at (0,-3){};
\coordinate (infl) at (0.005,-2.203){};
\coordinate (infl2) at (0,-2.293){};
\foreach \X in {a,b,ad,add}
\node at (\X){$\bullet$};
\draw[thick,color=blue] (a) arc (-37.57:79.73:0.82);
\draw[thick,color=blue] (a) arc (142.43:207.22:0.734);
\draw[thick,color=blue] (infl) arc (-6.59:27.22:0.734);
\draw[thick,color=blue] (infl) arc (173.41:180.55:0.734);
\draw[thick,color=blue] (infl2) to (add);
\end{scope}
}
\foreach \b in {0,1}{
\begin{scope}[xscale=\b*2-1,yscale=-\b*2+1,xshift=\b*1cm,yshift=\b*-0.614cm]
\coordinate (a) at (0,-1){};
\coordinate (add) at (0,-3){};
\coordinate (infl) at (0.005,-2.203){};
\coordinate (infl2) at (0,-2.293){};
\draw[thick,color=red] (a) arc (-37.57:79.73:0.82);
\draw[thick,color=red] (a) arc (142.43:207.22:0.734);
\draw[thick,color=red] (infl) arc (-6.59:27.22:0.734);
\draw[thick,color=red] (infl) arc (173.41:180.55:0.734);
\draw[thick,color=red] (infl2) to (add);
\end{scope}
}
\end{scope}
}
\begin{scope}[xshift=9.075 cm,yscale=-1,yshift=-0.614cm]
\node at (0,-1){$\bullet$};
\node at (0,0){$\bullet$};
\node at (0,-2){$\bullet$};
\node at (0,-3){$\bullet$};
\end{scope}
\end{tikzpicture}
\caption{Together with the black semicircles, the blue and red curves both give a valid reconstruction under the $0.72$-sampling condition.}
\label{Fig:cex4}
\end{subfigure}
\caption{\label{Fig:cex}}
\end{figure}
Now we want to extend this construction. See \cref{Fig:cexz} for what follows. We add a point $e$ such that $d$ is the midpoint between $c$ and $e$.
Next, we put a circle $S_3$ with radius $\frac{d_r}{0.72}$ so that it is tangent to $S_2$ at $d$, and a circle $S_4$ with the same radius as $S_3$ tangent to $S_3$ such that $e$ lies on $S_4$.
If we extend $\C$ such that it contains $[d,e]$ along $S_3$ and $S_4$ in the obvious way, then $\{a,b,c,d,e\}$ is a $0.72$-sample of $\C$.
To see this, note that if $s$ is the midpoint of $[d,e]$, then the difference in $x$-coordinate between $s$ and $d$ is less than that between $r$ and $d$, so $d_s<d_r$. The closest points on the medial axis to $s$ are the centers of $S_3$ and $S_4$, which have a distance of $\frac{d_r}{0.72}>\frac{d_s}{0.72}$ to $s$.

The tangent of $\C$ at $d$ is much closer to being vertical than the tangent at $e$, and if we add another point $f$ such that $e$ is the midpoint between $d$ and $f$, then we can extend $\C$ to $f$ similarly to how we extended $\C$ from $d$ to $e$ in such a way that $\{a,b,c,d,e,f\}$ is a $0.72$-sample, and such that the tangent of $\C$ at $f$ is vertical.

We can do the same below $a$, adding two points such that $\C$ can be extended downwards and the tangent of $\C$ at the lowest point is vertical. Now do the same for $\C'$, and add a sequence of points densely sampling a semicircle to connect $\C$ and $\C'$ as shown in \cref{Fig:cex3}. Again, the points shown make up a $0.72$-sample of the curve.
Next, we put many copies of this construction next to each other as shown in \cref{Fig:cex4}. Each copy is translated horizontally such that $d(b,c')$ is equal to the distance between $b'$ in one copy and $c$ in the copy on its right. If we ignore what happens to the far right or left, there are two ways to draw a set of curves with endpoints among the bottom points such that the set of points is a $0.72$-sample of the union of curves.

We now take this long strip of points and curves and bend it slightly upwards such that they are contained in an annulus and the ends meet; see \cref{Fig:cex5}.
As the length of this strip goes to infinity, the distances from points on the curve to the closest sample point and the medial axis are distorted by a factor that approaches $1$ when we bend it into the annulus.
Our arguments for the the set of points being a $0.72$-sample works equally well for an $\epsilon>0.72$ sufficiently close to $0.72$, so after turning the (sufficiently long) strip into an annulus, the point set stays a $\delta$-sample for some $\delta>0.72$.

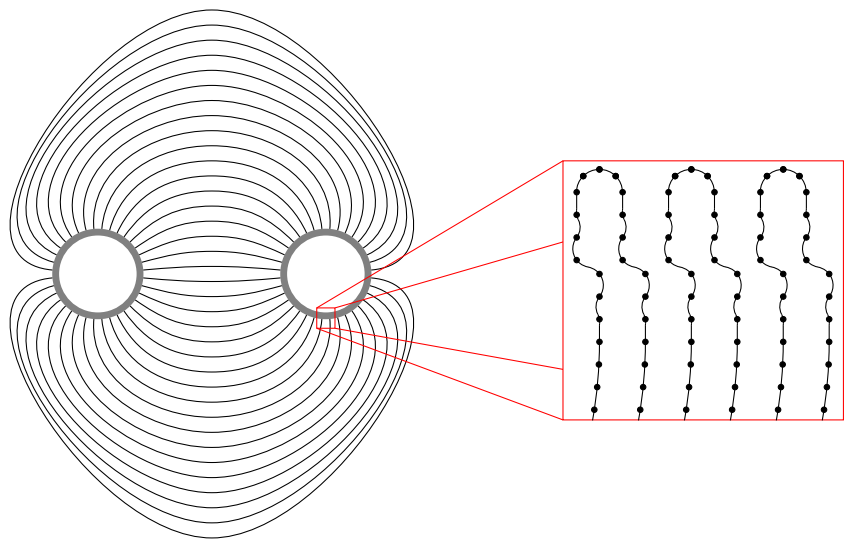
\begin{figure}
\centering
\begin{tikzpicture}[scale=.3]
\fill [gray,even odd rule] (0,0) circle[radius=2cm] circle[radius=1.7cm];
\fill [gray,even odd rule] (10,0) circle[radius=2cm] circle[radius=1.7cm];
\foreach \a in {5,15,...,175}{
\coordinate (x) at ({-2*cos(\a)+10},{2*sin(\a)}){};
\draw ({2*cos(\a)},{2*sin(\a)}) to [out=\a,in=180] (5,\a/15) to [out=0,in=180-\a] (x);
}
\foreach \a in {185,195,...,355}{
\coordinate (x) at ({-2*cos(\a)+10},{2*sin(\a)}){};
\draw ({2*cos(\a)},{2*sin(\a)}) to [out=\a,in=180] (5,\a/15-24) to [out=0,in=180-\a] (x);
}
\draw[color=red] (9.6,-2.4) rectangle (10.4,-1.5);
\draw[color=red] (9.6,-1.5) to (20.4,5);
\draw[color=red] (9.6,-2.4) to (20.4,-6.45);
\begin{scope}
\clip(10.4,-6.45) rectangle (20.4,5);
\draw[color=red] (10.4,-1.5) to (32.7,5);
\draw[color=red] (10.4,-2.4) to (32.7,-6.45);
\end{scope}
\begin{scope}[xshift=22cm]
\draw[color=red] (-1.6,-6.45) rectangle (10.7,5);
\def\radius{0.12};
\foreach \a in {0,1,2}{
\begin{scope}[xshift=\a*4.03cm]
\draw[color=black,fill=black] (0.712,4.326) circle (\radius);
\draw[color=black,fill=black] (0,4.621) circle (\radius);
\draw[color=black,fill=black] (-0.712,4.326) circle (\radius);
\draw[color=black,fill=black] (0.712,4.326) circle (\radius);
\draw[color=black,fill=black] (0,4.621) circle (\radius);
\draw[color=black,fill=black] (-0.712,4.326) circle (\radius);
\draw (1.008,3.614) arc (0:180:1.008);
\end{scope}
}
\foreach \a in {0,1,...,5}{
\begin{scope}[xshift=\a*2.015 cm]
\draw (0,-3) arc (0:-10:20);
\draw[color=black,fill=black] (-0.025,-4) circle (\radius);
\draw[color=black,fill=black] (-0.1,-5) circle (\radius);
\draw[color=black,fill=black] (-0.226,-6) circle (\radius);
\foreach \b in {0,1}{
\begin{scope}[yscale=-\b*2+1,xscale=-\b*2+1,xshift=\b*1cm,yshift=\b*-0.614cm]
\coordinate (a) at (0,-1){};
\coordinate (b) at (0,0){};
\coordinate (ad) at (0,-2){};
\coordinate (add) at (0,-3){};
\coordinate (infl) at (0.005,-2.203){};
\coordinate (infl2) at (0,-2.293){};
\foreach \X in {a,b,ad,add}
\draw[color=black,fill=black] (\X) circle (\radius);
\draw (a) arc (-37.57:79.73:0.82);
\draw (a) arc (142.43:207.22:0.734);
\draw (infl) arc (-6.59:27.22:0.734);
\draw (infl) arc (173.41:180.55:0.734);
\draw (infl2) to (add);
\end{scope}
}
\end{scope}
}
\end{scope}
\end{tikzpicture}
\caption{A reconstruction of the whole point set with the two annuli in grey. One can make sure that the curve is connected, and the point set is a $0.72$-sample of it.}
\label{Fig:cex5}
\end{figure}

Finally, we consider two such annuli with ``the ends tied together'', meaning that we draw curves between endpoints in the first annulus and endpoints in the second annulus, and sample the curves densely; see \cref{Fig:cex5}. In each of the two annuli, we have two choices of how to draw the curve, as illustrated in \cref{Fig:cex4}, which gives four different choices.
Exactly two of these choices result in a connected curve, and in all four cases, the set of points is a $0.72$-sample of the curve or union of curves. Summing up, we get \cref{Thm:counterex}.

\section{Counterexample to hypersurface reconstruction for $\epsilon=0.72$}
\label{Sec:hypersurface}

We have not defined what ``correct reconstruction'' means in higher dimensions. But assuming that preserving the number of connected components is required, we show that correct reconstruction of hypersurfaces in $\R^d$ using $0.72$-samples is impossible for any $d\geq 2$.
\begin{theorem}
For any $d\geq 2$, there is a finite point set $\Ss\subset \R^d$ that is a $0.72$-sample of two manifolds $\C$ and $\C'$ without boundary of dimension $d-1$ with a different number of connected components.
\end{theorem}
\begin{proof}
The case $d=2$ follows immediately from \cref{Thm:counterex}. For any point $p=(x,y)\in (0,\infty)\times \R$, let $p^\circ$ be the circle centered at $(0,y)$ containing $p$.
For any set $X\subset (0,\infty)\times \R$, let $X^\circ = \bigcup_{p\in X}p^\circ$.
Let $\C_i$ be as in \cref{Thm:counterex} for $1\leq i \leq 4$, and let $\Ss_\text{curve}$ be the $0.72$-sample as constructed in the previous section.
Pick a constant $R$ and translate $\C_i$ so that it is contained in $(R,\infty)\times \R$.
Similarly to how we bent a strip into a large annulus earlier, by choosing $R$ large, we can make sure that a sufficiently dense subset $\Ss$ of $\Ss_\text{curve}^\circ$ is a $\delta$-sample of $\C_i$ for some $\delta>0.72$.
Choosing $i=1$ and $i=3$, the theorem for $d=3$ follows.
To get the theorem for larger $d$, one can iterate the construction we used to get from $d=2$ to $d=3$.
\end{proof}

\section{Discussion}
\label{Sec:discussion}

We have only considered unions of closed curves. An obvious question is if our work generalizes to open curves. We expect that this can be dealt with by a slight tweak of the algorithms when the endpoints are far apart:
Instead of immediately connecting a point to its ``correct'' neighbors (i.e., its closest and closest ``compatible'' neighbors), one should add an edge between two points only when both points consider the other as a ``correct'' neighbor. However, we have not tried to turn this intuition into a precise statement.

Though this paper is mainly about curve reconstruction, we hope that it can also be a step towards improving the sampling conditions for surface reconstruction. Our arguments are valid for samples in any ambient dimension, and we expect many of our intermediate results to carry over to points on surfaces instead of curves. We consider generalizing our approach to surface reconstruction to be a promising direction of future research.

\bibliography{curve2}

\appendix

\section{Technical proofs}

\subsection{Proof of \cref{Lem:Xab}}
\label{Subsec:proof_Xab}

Recall the lemma:
\begin{Xab}
Let $a\to b$.
\begin{itemize}
\item[(i)] Let $p'\in E(a,b)\cap \partial U(a,b)$, and let $x$ be the point in $X(a,b)$ maximizing the distance to $p'$. Then $d(p',a)=\epsilon d(p',x)$, $d(p',x)=d(a,x)$ and $2\angle axp' = \angle axb$.
\item[(ii)] Let $p$ be the midpoint of $[a,b]$. Then $X(a,b)\subset B_p\left(\frac{d_p}{\epsilon} \right)$.
\item[(iii)] $(a,b)\subset U(a,b)$.
\end{itemize}
\end{Xab}
\begin{proof}
(i): See \cref{Fig:pythagoras} for the constructions that follow.
Let $\Pi$ be the plane spanned by $a,b,p'$. We use the notation $\Pi^+$ for the closed upper half plane and $\Pi^-$ for the closed lower half plane.
Without loss of generality, we can assume $a=(-1,0)$, $b=(1,0)$ in this plane, and that $p'\in \Pi^-$.
For all $y\in X(a,b)$, $d((0,0),y)$ is fixed.
Since $d(p,y)\leq d(p,(0,0))+d((0,0),y)$ and this inequality is only strict if $y$ lies on the nonnegative part of the $y$-axis, we get that $x\in X(a,b)\cap \Pi^+$ for $x$ as defined in the lemma.
Thus, $\overline{xp'}$ contains the midpoint of $\overline{ab}$, so $2\angle axp' = \angle axb$.
Since $p'$ is on the boundary of $U(a,b)$, $d(p',x)=d(a,x)$. Elementary calculation using the Pythagorean theorem shows that $\epsilon d(p',x)=d(p',a)$, proving (i).
\begin{figure}
\centering
\begin{tikzpicture}[scale=3]
\coordinate (a) at (-1.257,0){};
\node[left] at (a){$a$};
\coordinate (b) at (1.257,0){};
\node[left] at (b){$b$};
\coordinate (x) at (0,1.555){};
\node[right] at (x){$x$};
\coordinate (p) at (0,-0.444){};
\node[right] at (p){$p'$};
\coordinate (o) at (0,0){};
\foreach \X in {a,b,x,p,o}
\node at (\X){$\bullet$};
\draw (p) to (x) to (a) to (p);
\draw (a) to (o);
\node at (-0.8,0.8){$2$};
\node at (-0.7,-0.3){$2\epsilon$};
\node at (0.1,-0.22){$\epsilon^2$};
\node at (0.2,0.78){$2-\epsilon^2$};
\node at (-0.6,0.1){$\epsilon\sqrt{4-\epsilon^2}$};
\draw (-0.05,0) to (-0.05,0.05) to (0,0.05);
\end{tikzpicture}
\caption{Distances in $\Pi$ scaled by a factor of $\epsilon\sqrt{4-\epsilon^2}$. All the given values follow from $x\in X(a,b)$ and $d(a,x)=d(p',x)$, and we see that $d(a,p')=\epsilon d(a,x)$. \label{Fig:pythagoras}}
\end{figure}
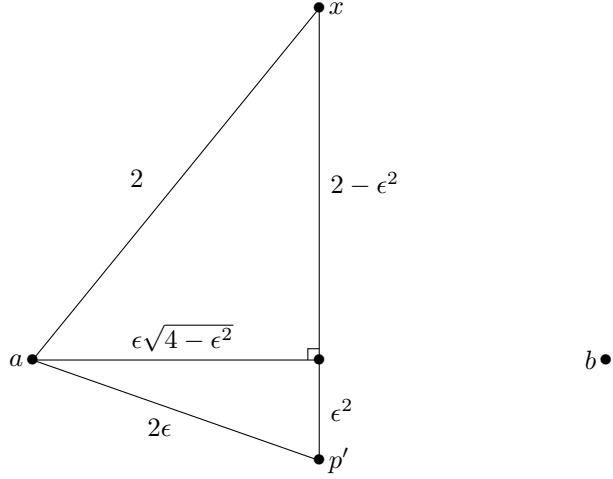

(ii): Let $\Pi$ be a plane spanned by $a,b,p$. Without loss of generality, we can assume $a=(-1,0)$ and $b=(1,0)$, and that $p\in \Pi^-$.
Then, for the same reasons as in the proof of (i), $x\in X(a,b)\cap \Pi^+$ maximizes the distance to $p$ among all points in $X(a,b)$.
It suffices to show that $x\in B_p\left(\frac{d_p}{\epsilon}\right)$. Let also $p'\in \Pi^-$ be as in the proof of (i).

Consider the case when $p$ is below or equal to $p'$. Let $x'$ be the point on the $y$-axis such that
$d(p,x')=d(a,x')$. Since $\triangle px'a$ and $\triangle p'xa$ are both isosceles and $\angle xp'a >\angle x'pa$, we have 
\[
\frac{d(p,x')}{d(p,a)}\leq \frac{d(p',x)}{d(p',a)}= \epsilon,
\]
so $x'\in B_p\left(\frac{d(p,a)}{\epsilon} \right)
=B_p\left(\frac{d_p}{\epsilon} \right)$.
Since $d(x',a)\leq d(x',p)$, we have a contradiction by \cref{lem:unique2} (ii).

Thus, we can assume that $p$ is above $p'$. Let $a'\in\overline{ax}$ be such that $\overline{p'a}$ is parallel to $\overline{pa'}$. Then \[\epsilon=\frac{d(p',a)}{d(p',x)}=\frac{d(p,a')}{d(p,x)}\leq \frac{d(p,a)}{d(p,x)},\] which proves (ii).

Since we know that $X(a,b)\subset B_p\left(\frac{d_p}{\epsilon} \right)$ from (ii), (iii) follows from \cref{lem:unique2} (ii).
\end{proof}

\subsection{Proof of \cref{Prop:main}}
\label{Subsec:proof_main}

Recall the statement of the proposition:
\begin{propmain}
Let $\epsilon\leq 0.765$, and let $a\to b\to c$ with $p\in (a,b)$ and $d(p,b)\leq d(p,a)$. Then for any $x$ such that $d(x,p)=d(x,b)=\frac{d_p}{\epsilon}$, $(b,c]\cap B_x\left(\frac{d_p}{\epsilon}\right) =\emptyset$.
\end{propmain}

\begin{proof}
\begin{figure}
\centering
\begin{tikzpicture}[scale=4]
\draw (0,0) arc (180:60:1.5);
\coordinate (p) at (0,0){};
\node[left] at (p){$p$};
\coordinate (b) at (0.3333,0.9428){};
\node[left] at (b){$b$};
\coordinate (x) at (1.5,0){};
\node[right] at (x){$x$};
\coordinate (n) at (0.333,0){};
\node[below] at (n){$n$};
\coordinate (c) at (1.655,1.492){};
\node[below] at (c){$c'$};
\coordinate (cc) at (2,1.5){};
\node[above] at (cc){$c$};
\coordinate (m) at (0.994,1.217){};
\node[above] at (m){$m$};
\coordinate (y) at (1.334,0.4){};
\node[right] at (y){$y$};
\coordinate (pp) at (0.9,1.444){};
\node[above] at (pp){$p'$};
\coordinate (C) at (2.2,1.55){};
\node[right] at (C){$\C$};
\foreach \X in {p,b,x,n,c,cc,m,y,pp}
\node at (\X){$\bullet$};
\draw (p) to (x);
\draw (p) to (y);
\draw (b) to (n);
\draw (b) to (x);
\draw (b) to (y);
\draw (b) to (c);
\draw (x) to (m);
\draw (0.948,1.198) to (0.967,1.152) to (1.0134,1.171);
\draw (0.383,0) to (0.383,0.05) to (0.333,0.05);
\draw[very thick] (p) to [out=75,in=240] (b) to [out=60,in=200] (pp) to [out=20,in=174.06] (c) to [out=-5.94,in=190] (cc) to [out=10,in=200] (C);
\end{tikzpicture}
\caption{Illustration of constructions in the proof of \cref{Prop:main}. \label{Fig:with_curve}}
\end{figure}
Let $Z$ be the set of points $x$ with $d(x,p) = d(x,b)=\frac{d(p,b)}{\epsilon}$, which is equal to $\frac{d_p}{\epsilon}$ by \cref{Lem:d_p_closest}.
We begin by observing that $\cl$ is well-defined on $Z$, as $Z\subset B_p\left(\frac{d_p}{\epsilon}\right)$, and that $\cl(x)\in (p,b)$ by \cref{lem:unique2} (i).

We now assume that there is a $c'\in (b,c]$ with $d(c',Z)\leq \frac{d_p}{\epsilon}$. We will derive a contradiction, which proves the proposition.
Now, $d(c',Z)$ varies continuously with $c'$, so either there is a $c'\in (b,c]$ with $d(c',Z) = \frac{d_p}{\epsilon}$, or $d(c',Z) < \frac{d_p}{\epsilon}$ for all $c'\in (b,c]$.
Note that the point in $Z$ closest to $c'$ lies in the plane spanned by $p,b,c'$. (The points $p,b,c'$ cannot lie on a line.)

In the case where $d(c',Z) < \frac{d_p}{\epsilon}$ for all $c'\in (b,c]$, pick a $c'$ close enough to $b$ that it lies in $B_p\left(\frac{d_p}{\epsilon}\right)$, and pick $x\in Z$ such that $d(x,c')<\frac{d_p}{\epsilon}$.
Since $x\in B_p\left(\frac{d_p}{\epsilon}\right)$, we get $\overline{xc'}\in B_p\left(\frac{d_p}{\epsilon}\right)$. Thus $\cl$ is defined on $\overline{xc'}$.
Clearly, $\cl(c')=c'$, and we have already observed that $\cl(x)\in (p,b)$.
For all $y\in \overline{xc'}$, we have both $d(b,y)> d(c',y)$ and $d(p,y)> d(c',y)$, so $b,p\notin \cl(\overline{xc'})$.
But this means that $\cl(\overline{xc'})$ is disconnected, which by continuity of $\cl$ is a contradiction.

Thus, we assume that there is a $c'\in (b,c]$ with $d(c',Z)=\frac{d_p}{\epsilon}$. As observed, we can pick $x\in Z$ such that $d(c',x)=\frac{d_p}{\epsilon}$ and $x$ lies in the plane spanned by $p,b,c'$. We recommend keeping an eye on \cref{Fig:with_curve} to follow the constructions ahead.

Draw the normal from $x$ to $\overline{bc'}$ and call the intersection $m$. Draw the normal from $b$ to $\overline{px}$ and call the intersection $n$.
Suppose $d(b,m)< d(b,n)$. Since the right triangles $\triangle nxb$ and $\triangle mbx$ share the hypotenuse $\overline{xb}$, the Pythagorean theorem tells us that $d(x,n)< d(x,m)$. Thus, there is a $y$ on $\overline{xm}$ such that
\[\frac{d(x,n)}{d(b,n)} =\frac{d(y,m)}{d(b,m)}.\]
If instead $d(b,m)\geq d(b,n)$, let $y=x$.

Let $p'$ be the midpoint on $[b,c']$. We claim that $d_{p'}=d(p',b)$, which is equivalent to $d(p',b)\leq d(p',c)$ by \cref{Lem:d_p_closest}. If $d(p',b)> d(p',c)$, the line segment $I$ from $p'$ to $c'$ is contained in $B_{p'}(d_{p'})\subset B_{p'}\left(\frac{d_{p'}}{\epsilon}\right)$, so $I$ does not intersect the medial axis. By \cref{Lem:I}, $[p',c)\subset B_{p'}(d(p',c))$, so $d(p',c)\geq d(p',c')=d(p',b)$, a contradiction.

We now state two claims, show how the proposition follows from them, and then we prove the claims.
\begin{claim}
\label{Cl:i}
$y\in B_p\left(\frac{d_p}{\epsilon}\right)$.
\end{claim}
\begin{claim}
\label{Cl:ii}
$y\in B_{p'}\left(\frac{d_{p'}}{\epsilon}\right)$.
\end{claim}
Let $B=B_p\left(\frac{d_p}{\epsilon}\right)$. We know that $x\in B$, so it follows from \cref{Cl:i} that $\overline{xy}\subset B$. In addition, $\overline{yp'}\subset B_{p'}\left(\frac{d_{p'}}{\epsilon}\right)$ by \cref{Cl:ii}. Thus, $\cl$ is defined on $I \coloneqq \overline{xy}\cup \overline{yp'}$.
We have observed that $\cl(x)\in (p,b)$, and clearly, $\cl(p')=p'$. Thus, by continuity of $\cl$, $\cl(I)$ contains either $b$ or $c'$. But for all $z\in I$, $d(b,z)=d(b,c')$, so $b,c'\notin \cl(I)$, a contradiction.

\begin{claimproof}[Proof of \cref{Cl:i}]
For the constructions ahead, see \cref{Fig:without_curve}. We already know that $x\in B_p\left(\frac{d_p}{\epsilon} \right)$, so it suffices to show that $d(p,y)\leq d(p,x)$. In fact, since $n$ is on the line segment $\overline{px}$, it is enough to show $d(n,y)\leq d(n,x)$, which will be our goal. This is a tautology in the case $d(b,m)\geq d(b,n)$, in which we defined $x=y$, so assume $d(b,m)< d(b,n)$.

Since $\angle xnb$ and $\angle xmb$ are right, $b$, $x$, $m$ and $n$ all lie on a circle $C$ where $\overline{xb}$ is a diameter. Fix $b$, $n$ and $x$, and let $m$ vary along $C$. By definition of $y$, $\angle myb = \angle nxb$ and thus $\angle xyb$ are fixed, so as $m$ moves along $C$, $y$ moves along a circle $C'$ containing $b$ and $x$.
Let $C''$ be the circle containing the midpoint between $x$ and $x'$ for all $x'\in C'$; in particular, the midpoint $q$ on $\overline{xy}$ and the midpoint $r$ on $\overline{xb}$ both lie on $C''$.

Letting $m$ approach $b$, so that $y$ approaches $b$ along $C'$, we see that $C$ and $C'$ meet at an angle of $\angle mby=\angle nbx$. Since $\overline{xb}$ is perpendicular to $C$ at $b$, this implies that $\overline{nb}$ is perpendicular to $C'$ at $b$.
Since $\angle xnb$ is right, $n$ is the midpoint between $x$ and another point on $C'$, so $n\in C''$. Thus $\angle xzn$ is constant for all $z$ on the arc of $C''$ between $x$ and $n$ containing $q$.
In \cref{Fig:pythagoras}, an isosceles triangle with the base length $\epsilon$ times the leg length is illustrated, which is thus similar to $\triangle pxb$.
Some calculation using the lengths given in \cref{Fig:pythagoras} shows that for $\epsilon\leq 0.765$, $d(x,n)\geq d(n,b)$. It follows that $\angle xqn=\angle xrn\geq 90^\circ$. This implies $d(n,y)\leq d(n,x)$, which was what we wanted to show.
\end{claimproof}
\begin{figure}
\centering
\begin{tikzpicture}[scale=5]
\coordinate (p) at (0,0){};
\node[left] at (p){$p$};
\coordinate (b) at (0.3333,0.9428){};
\node[left] at (b){$b$};
\coordinate (r) at (0.9167,0.4714){};
\node[below] at (r){$r$};
\coordinate (x) at (1.5,0){};
\node[right] at (x){$x$};
\coordinate (n) at (0.3333,0){};
\node[below] at (n){$n$};
\coordinate (m) at (0.994,1.217){};
\node[above] at (m){$m$};
\coordinate (y) at (1.334,0.4){};
\node[right] at (y){$y$};
\coordinate (q) at (1.417,0.2){};
\node[right] at (q){$q$};
\foreach \X in {p,b,r,x,n,m,y,q}
\node at (\X){$\bullet$};
\draw (p) to (x) to (m) to (b) to (x);
\draw (n) to (b) to (y);
\centerarc[black](y)(151:112.5:0.12);
\draw (1.38,0) arc (180:141.06:0.12);
\draw (0.948,1.198) to (0.967,1.152) to (1.0134,1.171);
\draw (0.383,0) to (0.383,0.05) to (0.333,0.05);
\draw (r) circle (0.75);
\centerarc[black](0.9171,-0.1247)(12.07:167.93:0.5961);
\node at (1.6,0.9){$C$};
\node at (0.56,0.42){$C''$};
\end{tikzpicture}
\caption{Illustration for the proof of \cref{Cl:i}. The circle $C'$ through $x$, $y$ and $b$ is left out to avoid cluttering the figure. \label{Fig:without_curve}}
\end{figure}
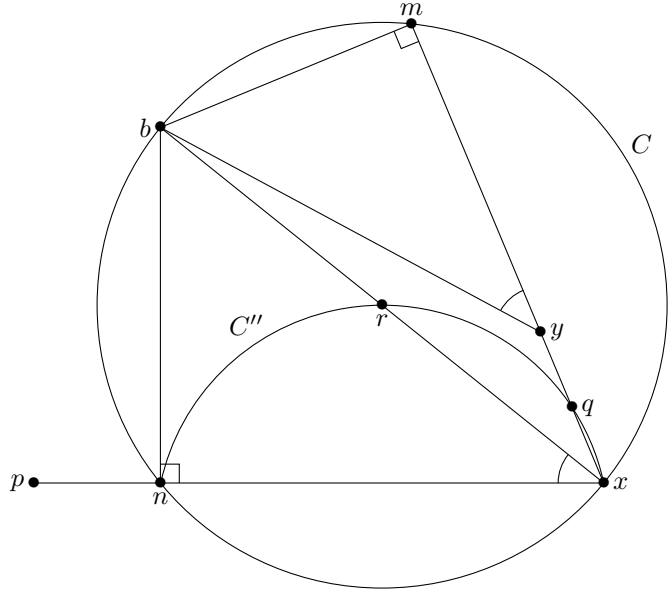

\begin{claimproof}[Proof of \cref{Cl:ii}]
Observe that $d(m,z)$ is the same for any $z\in X(b,c')$, since $m$ is the midpoint of $\overline{bc'}$. We want to show that $d(m,y)\leq d(m,z)$ for such a $z$. This implies that $y$ is in the convex hull of $X(b,c')$, so since $X(b,c')\subset B_{p'}\left(\frac{d(p',b)}{\epsilon}\right)$ by \cref{Lem:Xab} (ii), this will prove the claim.

For any $s\neq t\in \R^d$ and $u\in X(s,t)$, the angle $\rho\coloneqq\frac{\angle sut}{2}$ depends only on $\epsilon$. We have $\rho = \angle bzm$.
From \cref{Lem:Xab} (i), it follows that in an isosceles triangle where the ratio between the base and a leg is $\epsilon$, the angle between the legs is $\rho$.
In particular, $\angle bxn=\angle bxp=\rho$. But $\angle bxn\leq \angle bym$ (with equality if $x\neq y$), so $\angle bzm\leq \angle bym$, and it follows that $d(m,y)\leq d(m,z)$.
\end{claimproof}
\end{proof}

\subsection{Proof of \cref{Lem:abDisc}}
\label{Subsec:proof_abdisc}

\begin{abdisc}
Let $\epsilon\leq 0.765$, and let $a\to b\to c$. Then $(b,c]\cap B_x(d(x,a))=\emptyset$ for all $x\in X(a,b)$.
\end{abdisc}
\begin{proof}
Suppose there is an $x\in X(a,b)$ and $c'\in (b,c]$ such that $c'\in B_x(d(x,a))$, and let $\Pi$ be the plane spanned by $a,b,c'$. We can assume that $x\in \Pi$, since the point in $X(a,b)$ minimizing the distance to $c'$ is in $\Pi$.

First assume $d(x,c')=d(x,b)$. Let $n$ be the midpoint on $\overline{ab}$, and define $y$ and $p'$ as in the proof of \cref{Prop:main} using $n$, $x$, $b$ and $c'$.
Let $p$ be the midpoint on $[a,b]$. If $y\in B\coloneqq B_p\left(\frac{d_p}{\epsilon}\right)$, we get a contradiction by the same argument as in the proof of \cref{Prop:main}.
Observe that if $d(p,b)$ is fixed, $d(p,y)$ is maximized if $p$ is on the line through $n$ and $x$ on the side of $n$ furthest from $x$, so it is sufficient to consider this case.
\cref{Lem:Xab} (ii) gives us $x\in B$, and now $y\in B$ can be proved in exactly the same way as \cref{Cl:i} in the proof of \cref{Prop:main}. (Again, inserting $\epsilon=0.765$ in \cref{Fig:pythagoras} gives $d(n,x)\geq d(n,b)$.)

If $(b,c]\cap B_x(d(x,a))\neq\emptyset$ for some $x\in X(a,b)$, but there is no $c'\in (b,c]$ such that $d(c',X(a,b))=d(b,X(a,b))$, then every $c'\in (b,c]$ is contained in the interior of $B_x(d(x,a))$ for some $x\in X(a,b)$; i.e., $d(x,c')<d(x,b)$.
In particular, we can pick $c'\in B$.
Now $\cl$ is defined on $\overline{xc'}$, as $x,c'\in B$, but $\cl$ never takes $a$ or $b$ as a value on $\overline{xc'}$, as $c'$ is always closer than $a$ and $b$.
But \cref{lem:unique2} (i) implies that $\cl(x)\in (a,b)$, so $\cl(\overline{xc'})$ is disconnected. This is a contradiction, as $\overline{xc'}$ is connected and $\cl$ continuous.
\end{proof}

\subsection{Proof of \cref{Prop:finish}}
\label{Subsec:proof_finish}

Recall the statement of the proposition:
\begin{propfinish}
Let $\epsilon\leq 0.66$. Let $b$ be a sample point, $a$ a closest sample point to $b$, and $c$ the closest sample point to $b$ such that $(a,b,c)$ is compatible. Then $a$, $b$ and $c$ are consecutive.
\end{propfinish}

\begin{proof}
We know that $a$ and $b$ are consecutive by \cref{Prop:closest}. Suppose that $d(b,c)\leq d(b,c')$, that $(a,b,c)$ are compatible, and that $a$, $b$ and $c'$ are consecutive for some $c'\neq c$. By \cref{Lem:abDisc}, $(a,b,c')$ is compatible.
We will show that this leads to a contradiction, proves the lemma by process of elimination.
This will involve finding a series of bounds on various angles. Consult \cref{Fig:finish} for some geometric intuition.

\begin{figure}
\centering
\begin{tikzpicture}[scale=1.7]
\coordinate (q) at (0,-1){};
\node[left] at (q){$q$};
\coordinate (b) at (0,0){};
\node[below left=.04cm of b] {$b$};
\coordinate (c) at (1.777, 0.9173){};
\node[right] at (c){$c$};
\coordinate (p) at (-0.455, 0.66){};
\node[above] at (p){$p$};
\coordinate (a) at (-1.184, 0.995){};
\node[above right=.04cm of a] {$a$};
\foreach \X in {p,b,c,q,a}
\draw[color=black,fill=black] (\X) circle (.036);
\draw (q) to (b) to (c);
\draw (q) arc (-19.27:90:1.5152);
\draw (c) arc (-24.16:-170:1.6057);
\node at (0.33, -0.5){$\partial U$};
\node at (1.2, 0){$\partial U'$};
\end{tikzpicture}
\caption{Some restrictions on the point set $\{a,b,c,p,q\}$. We have $a,p\notin U$, where $U$ is the union of balls $B_x\left(\frac{d_q}{\epsilon}\right)$ as in \cref{Prop:main}, and $a\notin U' =\bigcup_{x\in X(b,c)}B_x(d(x,b))$. \cref{Cor:141} and \cref{Cor:102} express these restrictions in formulas. \label{Fig:finish}}
\end{figure}
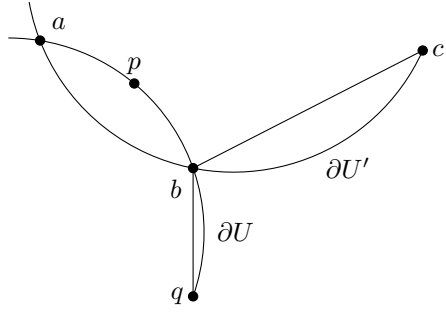

Let $q$ be the midpoint of $[b,c']$. Note that we get $d(c,b)\leq 2d(q,b)$ from $d(c,b)\leq d(c',b)$. Suppose that $d(a,b)\leq 0.773d(c,b)$, so $d(a,b)\leq 1.546d(q,b)$. Then by \cref{Cor:141},
\[
\angle qba > 70.73^\circ + \arccos(0.33\cdot 1.546) > 130.05^\circ,
\]
and by \cref{Cor:102},
\[
\angle cba > 51.45^\circ + \arccos(0.6231\cdot 0.773) > 112.65^\circ.
\]
From \cref{Lem:117}, we know that $\angle qbc> 117.3^\circ$. Thus,
\[
\angle qba + \angle cba + \angle qbc > 130.05^\circ + 112.65^\circ + 117.3^\circ =360^\circ,
\]
which is impossible. We conclude that $d(a,b)> 0.773d(c,b)$. This implies that there is a point $p\in (a,b)$ with $d(p,b)=0.4d(c,b)$, and $d(p,b)\leq 0.8d(q,b)$ then follows from $d(c,b)\leq 2d(q,b)$. By \cref{Cor:141},
\[
\angle qbp > 70.73^\circ + \arccos(0.33\cdot 0.8) > 145.42^\circ.
\]

\begin{figure}
\centering
\begin{tikzpicture}[scale=3]
\coordinate (n) at (0,0){};
\node[below] at (n){$n$};
\coordinate (p) at (0,.5){};
\node[left] at (p){$p$};
\coordinate (b) at (1.6,0){};
\node[right] at (b){$b$};
\coordinate (a) at (-1.4,0){};
\node[above] at (a){$a$};
\foreach \X in {n,p,b,a}
\draw[color=black,fill=black] (\X) circle (.02);
\draw (n) to (p) to (b) to (a);
\node at (.9,0.35){$d(p,b)$};
\node at (.8,-.15){$\geq \frac{d(a,b)}{2}$};
\draw (0,.07) to (.07,.07) to (.07,0);
\end{tikzpicture}
\caption{If $d(p,a)\leq d(p,b)$ and $n$ is where the normal from $p$ meets $\overline{ab}$, then $2d(n,b)\geq d(a,b)$ and $\cos(\angle pba) = \frac{d(n,b)}{d(p,b)} \geq \frac{d(a,b)}{2 d(p,b)}$. \label{Fig:cosine}}
\end{figure}
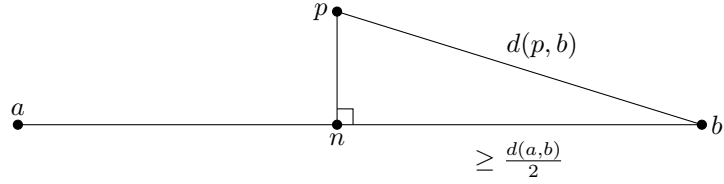

Suppose $d(p,a)\leq d(p,b)$. Then \cref{Fig:cosine} gives a simple geometric proof that
\[
\cos(\angle pba) \geq \frac{d(a,b)}{2d(p,b)}=\frac{d(a,b)}{0.8d(c,b)},
\]
which is equivalent to
\begin{equation}
\label{Eq:pba}
\angle pba \leq \arccos\left(1.25\frac{d(a,b)}{d(c,b)}\right).
\end{equation}
In addition, we have
\begin{align*}
\angle pba &\geq \angle cba - \angle cbp\\
&\geq \angle cba - (360^\circ - \angle qbp - \angle qbc)\\
&> 51.45^\circ + \arccos\left(0.6231\frac{d(a,b)}{d(c,b)}\right) + 145.42^\circ + 117.3^\circ - 360^\circ,
\end{align*}
where we have again used \cref{Cor:102} and \cref{Lem:117} to get bounds on $\angle cba$ and $\angle qbc$, respectively. Putting this together with \cref{Eq:pba}, we get
\[
\arccos\left(1.25\frac{d(a,b)}{d(c,b)}\right) > \arccos\left(0.6231\frac{d(a,b)}{d(c,b)}\right) - 45.83^\circ.
\]
This is false for $d(a,b)= 0.773d(c,b)$, and therefore false whenever $d(a,b)\geq 0.773d(c,b)$, as the derivative of the left hand side is less than the derivative of the right hand side. Thus, we have a contradiction, so $d(p,a)> d(p,b)$.

Observe that
\[
\angle pbc\leq 360^\circ -\angle qbp - \angle qbc < 360^\circ -145.4^\circ -117.3^\circ = 97.3^\circ.
\]

Let $x$ be an element in $X(a,b)$ maximizing the distance from $p$. By \cref{Lem:Xab} (iii), we then have (following \cref{Fig:pythagoras})
\[
\angle abp = \angle xbp - \angle xba < 90^\circ - \arcsin\left(\frac{2-0.66^2}{2}\right) < 39^\circ,
\]
so by \cref{Cor:102},
\[
\angle pbc \geq \angle cba - \angle abp > 102.9^\circ - 39^\circ = 63.9^\circ.
\]
The remainder of the proof is similar to the proof of \cref{Lem:117}, but with the simplifications that we have the equality $d(p,b)=0.4d(c,b)$ instead of an inequality, and that we only have to consider $\angle pbc$ in the narrow range $[63.9^\circ,97.3^\circ]$, which saves us some work.
We restrict ourselves to the plane spanned by $p$, $b$ and $c$. Let $p=(0,-1)$, $b=(0,0)$, let $c$ be in the right half-plane, and let $s$ and $s'$ be the centers of the circles through $p$ and $b$ with radius $\frac{1}{0.66}$ such that $s$ is in the left half-plane and $s'$ is in the right.
One can calculate that if $\angle pbc$ is $63.9^\circ$ or $97.3^\circ$, then $d(s',c)< d(s',s)-\frac{1}{0.66}$. This means that $\cl(s')\notin [b,p]$, since by \cref{lem:unique2} (ii), $[p,b]\subset B_s\left(\frac{1}{0.66}\right)$.
But by \cref{lem:unique2} (i), $\cl(s')\in (p,b)$, so we get a contradiction. It follows that we get the same contradiction for all $\angle pbc\in [63.9^\circ,97.3^\circ]$, and we are done.
\end{proof}

\section{The relationship between $\epsilon$-sampling and $\rho$-sampling}
\label{Sec:rho-sampling}

In this section, we discuss two examples showing relative strengths and weaknesses of $\epsilon$-sampling and the $\rho$-sampling condition introduced in \cite{ohrhallinger2016curve}. To avoid ambiguity, we will refer to $\rho$-sampling as \emph{$\rho$-reach-sampling}.
\begin{definition}
For $a,b$ in the same connected component of $\C$, let
\[
\reach([a,b]) = \inf_{p\in [a,b]} \lfs(p).
\]
For $\rho>0$, a curve $\C$ is $\rho$-reach-sampled if for all $a\to b$ and $p\in [a,b]$, $d(p,\Ss)< \rho\reach([a,b])$.
\end{definition}
In \cite[Theorem 2]{ohrhallinger2016curve}, it is proved that for $\rho<0.9$, one can reconstruct $\rho$-reach-sampled curves in the plane. To judge the strength of the $\rho$-reach-sampling condition, the authors consider sampling parallel lines, among other examples.
In this case, $\lfs(p)$ is constant, so $\Ss$ is an $r$-sample if and only if it is an $r$-reach-sample.
Thus, in this case, requiring that $\Ss$ is a $0.66$-sample (or even a $0.72$-sample) is stricter than requiring that it is a $0.9$-reach-sample, which results in denser sampling.
Thus, $\rho$-reach-sampling is arguably a better sampling condition than $\epsilon$-sampling in this case.

\begin{figure}
\centering
\begin{tikzpicture}[scale=2]
\draw[thick] (-2,0) to (0,0);
\draw[thick] (1,0) to (2,0);
\centerarc[black, thick](2,-.5)(50:90:.5);
\draw[color=blue, thick] (2,-.5) to [out=250,in=90] (1.95,-1);
\draw[color=blue, thick] (0,-.02) to (0,-1);
\draw[color=red] (-.05,-.05) rectangle (.05,.05);
\draw[color=red] (-.05,.05) to (1,1);
\draw[color=red] (-.05,-.05) to (1,-1);
\begin{scope}
\clip(-1,-1) rectangle (1,1);
\draw[color=red] (.05,.05) to (3,1);
\draw[color=red] (.05,-.05) to (3,-1);
\end{scope}
\begin{scope}
\draw[color=red] (1,-1) rectangle (3,1);
\end{scope}
\end{tikzpicture}
\caption{The hook example. $\C$ is black; the medial axis is blue. \label{Fig:hook}}
\end{figure}
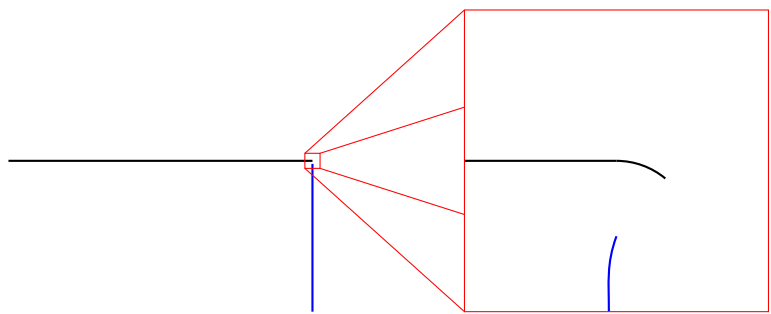
On the other hand, consider the following example, where we take the liberty of skipping some details and working with an open curve for the purpose of simplicity. This curve $\C$ is the union of the line segment $I$ from $a=(0,0)$ to $b=(1,0)$ and a tiny ``hook'' attached at $b$; see \cref{Fig:hook}.
We demand that $a\in \Ss$, and try to sample $I$ as sparsely as possible, first requiring that $\Ss$ is a $0.66$-sample, and then requiring that it is a $0.9$-reach-sample.

The medial axis of $\C$ is very close to the ray from $b$ going straight down. In particular, for $p\in \C$ that are not very close to $b$, $\lfs(p)\approx d(p,b)$.
We begin by deciding where to put the $p=(x,0)\in \Ss$ that is closest to $a$. Some calculation shows that if $x\leq 0.795$ and $q=\left(\frac{x}{2},0\right)$, then $0.66 d(q,b)>d(q,p)$, and it follows that the condition for $0.66$-sampling is satisfied for all $q\in [a,b]$ (assuming that $\lfs(q)$ is sufficiently close to $d(q,b)$).
However, if $x\geq 0.643$, then $d(q,p)>0.9d(p,b)\approx 0.9\lfs(p)$, so the $0.9$-reach-sampling condition is not satisfied.

The argument can be repeated for the next sample point after $p$ and so on until the distance to $b$ is not a good approximation for the local feature size. The conclusion is that we need more sample points for $\Ss$ to be a $0.9$-reach-sample than for $\Ss$ to be a $0.66$-sample.

The message we suggest to take home from the examples of parallel lines and the ``hook'' is that $\rho$-reach-sampling does well when the curvature and local feature size are fairly constant, while $\epsilon$-sampling does well when the curvature changes rapidly.
Which of the sampling conditions does better in practical applications is an open question that we will not attempt to answer here.

\end{document}